\documentclass[USenglish]{lipics-v2018}

\nolinenumbers

\title{The $k$-Fréchet distance}

\author{Hugo A. Akitaya}{Department of Computer Science, Tufts University{[Massachusetts, USA]}}{hugo.alves\_akitaya@tufts.edu}{}{supported by NSF awards CCF-1422311 and CCF-1423615, and the Science Without Borders scholarship program}

\author{Maike Buchin}{Faculty of Computer Science, TU Dortmund{[Germany]}}{maike.buchin@tu-dortmund.de}{}{}

\author{Leonie Ryvkin}{Department of Mathematics, Ruhr Universität Bochum{[Germany]}}{leonie.ryvkin@rub.de}{}{}

\author{Jérôme Urhausen}{Department of Information and Computing Sciences, Universiteit Utrecht{[Netherlands]}}{J.E.Urhausen@uu.nl}{}{supported by the Netherlands Organisation for Scientific Research under project 612.001.651}

\authorrunning{H. Akitaya, M. Buchin, L. Ryvkin, J. Urhausen}

\Copyright{Hugo A. Akitaya, Maike Buchin, Leonie Ryvkin, Jérôme Urhausen}

\subjclass{
• Theory of computation $\rightarrow$ Computational geometry  \\
• Theory of computation $\rightarrow$ Approximation algorithms analysis   \\
• Theory of computation $\rightarrow$ Parameterized complexity and exact algorithms}


\acknowledgements{We would like to thank Erik Demaine for contributing the key idea for proving hardness in the free space diagram in Section~\ref{sec:boxes}, as well as the organizers and other participants of the Intensive Research Program in Discrete, Combinatorial and Computational Geometry in Barcelona, 2018, for providing the perfect environment to meet other researchers.
}

\newcommand{\eps}{\varepsilon}
\newcommand{\R}{\mathbb{R}}

\hideLIPIcs
\begin{document}

\maketitle

\begin{abstract}
We introduce a new distance measure for comparing polygonal chains: the $k$-Fréchet distance. As the name implies, it is closely related to the well-studied Fréchet distance but detects similarities between curves that resemble each other only piecewise. The parameter $k$ denotes the number of subcurves into which we divide the input curves. 
The $k$-Fréchet distance provides a nice transition between (weak) Fréchet distance and Hausdorff distance.
However, we show that deciding this distance measure turns out to be NP-complete, which is interesting since both (weak) Fréchet and Hausdorff distance are computable in polynomial time.
Nevertheless, we give several possibilities to deal with the hardness of the $k$-Fréchet distance: besides an exponential-time algorithm for the general case, we give a polynomial-time algorithm for $k=2$, i.e., we ask that we subdivide our input curves into two subcurves each. 
We also present an approximation algorithm that outputs a number of subcurves of at most twice the optimal size.
Finally, we give an FPT algorithm using parameters $k$ (the number of allowed subcurves) and $z$ (the number of segments of one curve that intersects the $\eps$-neighborhood of a point on the other curve). 
\end{abstract}

\section{Introduction}
\label{sec:intro}
During the last decades, several methods for comparing geometrical shapes have been studied in a variety of applications,
for example analyzing geographic data, such as trajectories, or comparing chemical structures, e.g., protein chains or human DNA.
The Fréchet distance has been well-studied in the past twenty years since it has proven to be helpful in several of the mentioned applications.
The Hausdorff distance, another similarity measure, has also been used in applications and can be computed more efficiently than the Fréchet distance. However, it provides us with less information by taking only the overall shape of curves into consideration, not how they are traversed.

We introduce the $k$-Fréchet distance as a distance measure in between Hausdorff and (weak) Fréchet distance.
This measure allows us to compare shapes consisting of several parts: we cover the input curves by at most $k$ (possibly overlapping) subcurves each and ask for a matching of the subcurves such that each pair of matched subcurves has at most weak Fréchet distance~$\eps$ (where $\eps >0$ is a given constant).

Thus the new measure allows us to find similarities between curves that need to be cut and reordered to be similar under the Fréchet distance. For instance this could be objects of rearranged pieces such as chemical structures or handwritten characters and symbols.
Or imagine trajectories of tourists visiting a number of sights in a city. If the $k$-Fréchet distance of two trajectories is small, the respective tourists used similar routes to get to the sights. Additionally, for small $k$ we can conclude that the tourists visited many sights in the same order.
Another example is displayed in Figure~\ref{fig:different-k}, where we compare three different variants of writing the letter k by hand. Note that we deal with disconnected curves by concatenating the respective subcurves. Of course, we can easily identify that all three of them are k's by using the Hausdorff distance to compare them to a ``generic'' k, but the $k$-Fréchet distance provides us with more information: the 2-Fréchet distance between the second and the third k is high, because the strokes are set differently. Those k's are unlikely to be written by the same person. The 3-Fréchet distance, however, is small, because the letter consists of at most 3 strokes in general.

	\begin{figure}[ht]
		\centering
		\includegraphics[scale=0.6]{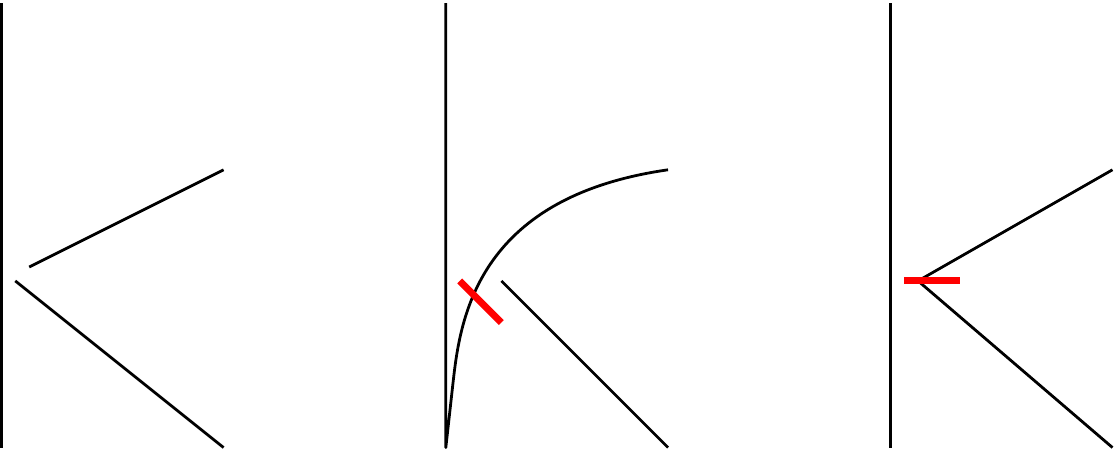}
		\caption{Three `k's written in a different way. For the middle and the right one, the 2-Fréchet distance is high and the 3-Fréchet distance is low.}
	\label{fig:different-k}
	\end{figure}

Characterizing the mentioned variants of the Fréchet distance next to the Hausdorff distance intuitively shows that the new distance measure bridges between weak Fréchet and Hausdorff distance.
As is common in the literature, we use the following analogy: we interpret our input curves as two paths, and imagine that these have to be traversed by a man and a dog, each of them walking on one of the paths.
For the (weak) Fréchet distance we ask for the length of a shortest leash so that man and dog can traverse their respective curves. In doing so they may choose their speeds independently. For the variant of the weak Fréchet distance, man and dog are also allowed to backtrack.

The Hausdorff distance finds for each point on either curve the closest point on the other curve, and takes the largest of the obtained distances. In terms of man and dog we do not need them to traverse the curves as such.
We just ask that for any fixed position on either path there is a position on the other one such that man and dog can stand on their respective positions using a leash of a fixed length. One could say they may ``jump'' on their curves any number of times as long as both man and dog can reach all positions on their respective curves without exceeding the given maximum distance, i.e., the leash length.
The $k$-Fréchet distance limits this number of jumps to a constant $k$ (actually, we have $k-1$ jumps), so we want man and dog to traverse their paths piecewise.
Note that we use the weak Fréchet distance as underlying distance measure, so we allow backtracking.

\textbf{Related work.}
Efficient algorithms were presented
for computing the Fréchet distance and the weak Fréchet distance
by Alt and Godau in 1995. There, they first introduced the concept of the free space diagram, which is key to computing this distance measure and its variants~\cite{altgodau}. Following their work, numerous variants and extensions have been considered. Here we mention only a few results related to our work.
Alt, Knauer and Wenk compared the Hausdorff to the Fréchet distance and discussed $\kappa$-bounded curves as a special input
instance~\cite{knauer}. In particular, they showed that for convex closed curves Hausdorff distance equals Fréchet distance.
For curves in one dimension Buchin et al.~\cite{walkdog} proved equality of Hausdorff and weak Fréchet distance using the
well-known Mountain climbing theorem~\cite{mountain}.
For computing the Hausdorff distance, Alt, Braß, Godau, Knauer and Wenk gave a thorough overview in \cite{hausdorff}.
Buchin~\cite{b-cfdts-07} gave the characterization of these measures in free space, which motivated our study of $k$-Fréchet distance.

For $c$-packed curves, Driemel, Har-Peled and Wenk presented a $(1+\varepsilon)$-approximation algorithm, which determines the Fréchet distance in near linear time~\cite{driemel}. An interesting variant was presented by Gheibi et al.: they studied the weak Fréchet distance but minimized the length of the subcurves on which backtracking is necessary \cite{gheibi}.
For general polygonal curves, Buchin et al.~\cite{Buchin2017} recently slightly improved the original algorithm of Alt and Godau, while Bringmann~\cite{Bringmann} showed that unless SETH fails no strongly subquadratic algorithm for the Fréchet distance exists.
Buchin, Buchin and Wang studied partial curve matching, where they presented a polynomial-time algorithm to compute the ``partial Fréchet similarity'' \cite{partialcurve}, and a variation of this similarity was presented by Scheffer in \cite{scheffer}. Also, Driemel and Har-Peled defined a Fréchet distance with shortcuts~\cite{shortcut}, which was proven to be the first NP-hard variant of the Fréchet distance in~\cite{shortcut2}.

Interestingly, both Hausdorff and (weak) Fréchet distance are computable in polynomial time. However, the $k$-Fréchet distance, as a distance measure that bridges between the two of them, proves to be NP-complete.

Buchin and Ryvkin first presented the $k$-Fréchet distance at EuroCG2018~\cite{eurocg}. However, there are in fact two variants of $k$-Fréchet distance: the cut and the cover variant. These differ in that the first asks to cut, i.e., partition the curves, whereas the latter asks to cover the curves with subcurves. In their previous work, Buchin and Ryvkin showed NP-hardness of the cut-variant and gave the 2-approximation for the cover variant. In this paper we study the cover variant, and discuss the cut variant only in the conclusion.

\textbf{Overview.} In the next chapter, we introduce and formally define the $k$-Fréchet distance. In Chapter~\ref{sec:cover-hard}, we determine its hardness in two steps: first we prove NP-hardness of a simpler auxiliary problem to gain some intuition (Section~\ref{sec:boxes}) for the then following reduction. We describe the construction in Section~\ref{sec:cover-hard-construction} and analyze its correctness in Section~\ref{sec:cover-hard-analysis}.
Finally, we present our algorithmic findings in Chapter~\ref{sec:cover-algo}. We give an XP-algorithm with parameter $k$, which even works in reasonable polynomial time for very small $k$. There is also an FPT algorithm using two parameters, again the selection size $k$ and the parameter $z$, which indicates how ``entangled'' the input curves are.
To complete the possible algorithmic approaches, we give a 2-approximation algorithm. 

\section{Preliminaries}
\label{sec:prelims}
First we define the \emph{Hausdorff distance}~\cite{knauer} for curves $P, Q \colon [0,1] \to \mathbb{R}^d$ as
	\[\delta_{\operatorname{H}}(P,Q) = \max(\tilde{\delta}_H(P,Q),\tilde{\delta}_H(Q,P)) \mbox{, where}\]	
\[\tilde{\delta}_{\operatorname{H}}(P,Q) = \max_{t_1 \in [0,1]} \min_{t_2 \in [0,1]}\Vert P(t_1)-Q(t_2)\Vert\]
denotes the directed Hausdorff distance from $P$ to $Q$. By  $\Vert \cdot \Vert$ we refer to the Euclidian norm in $\mathbb{R}^d$.

Now recall the \emph{Fréchet distance}~\cite{altgodau}:
For curves $P,Q \colon [0,1] \to \mathbb{R}^d$ it is given by
	\[\delta_{\operatorname{F}}(P,Q) = \inf_{\sigma} \max_{t \in [0,1]} \Vert P(t) - Q(\sigma(t)) \Vert ,\]
where the reparametrisations $\sigma \colon [0,1] \to [0,1]$ range over all orientation-preserving homeomorphisms.
A variant is the \emph{weak Fréchet distance} $\delta_{\operatorname{wF}}$ where both curves are reparameterised by $\sigma$ and $\tau$,
respectively, which range over all continuous surjective functions.

The Fréchet distance is typically illustrated by a man and a dog walking on the two curves where both may choose their speed
independently. For the Fréchet distance man and dog may not backtrack, whereas for the weak Fréchet distance they may.
The (weak) Fréchet distance corresponds to the shortest leash length allowing them to traverse the curves.

A well-known characterisation, which is key to efficient algorithms for computing both weak and (strong) Fréchet distance uses the free space diagram, which was introduced by Alt and Godau~\cite{altgodau}. First we recall the free~space $F_\varepsilon$:
	\[F_\varepsilon(P,Q)=\{(t_1,t_2) \in [0,1]^2 \colon \Vert P(t_1)-Q(t_2)\Vert \leq \varepsilon \}.\]
For piecewise-linear $P$ and $Q$, the free space diagram puts this information into an $(n \times m)$-grid where $n$ and $m$ are the numbers of segments in $P$ and $Q$ respectively. For the rest of this paper we assume that $m = \mathcal{O}(n)$ to simplify runtime expressions.

The Fréchet distance of two curves is at most a given value $\varepsilon$ if there exists a monotone path through the free space connecting the bottom left to the top right corner.
For the weak Fréchet distance to equal at most $\varepsilon$ such a path need not be monotone.
It may also start and end somewhere other than the corners of the diagram, as long as it touches all four boundaries.

We define further terms connected to the free space diagram below:
A \emph{component} of a free space diagram is a connected subset $c\subseteq F_\varepsilon(P,Q)$.
A set $S$ of components \emph{covers} a set $I\subseteq [0,1]_P$ of the parameter space (corresponding to the curve $P$) if $I$ is a subset of the projection of $S$ onto said parameter space, i.e., $\forall x\in I \colon \exists c\in S, y\in [0,1]_Q \colon (x,y)\in c$.
Covering on the second parameter space is defined analogously.
This means the weak Fréchet distance is smaller than $\eps$ if there is one component in $F_\varepsilon(P,Q)$ that covers both parameter spaces. 
Similarly, the Hausdorff distance can be tested by checking whether the set of all components covers both parameter spaces.
In this paper we extend this concept to also account for the number of components needed to cover the parameter spaces.

We define the \emph{$k$-Fréchet distance} $\delta_{\operatorname{kF}}(P, Q)$ as the minimal $\eps$ such that there is a set of at most $k$ components of $F_\eps(P,Q)$ that covers both parameter spaces.
That is, we cover the curves $P$ and $Q$ by at most $k$ pieces (i.e., subcurves) such that there is a matching of the pieces where two matched subcurves have small weak Fréchet distance. Note that we do not demand that the subcurves are disjoint.
In the man and dog analogy, we allow man and dog to ``jump'' on their respective curves, i.e., they may skip parts of their paths and come back later. We still ask for a complete traversal, but some parts of the curves may be traversed multiple times with jumps in between.

The decision problem for this distance measure asks whether for a fixed value of $k$, $\delta_{\operatorname{kF}}(P, Q)$ is smaller than or equal to a given $\eps$.
Naturally, for a fixed real $\varepsilon>0$, we would like to cut the curves into as few subcurves as possible (optimization version).
By definition, the $k$-Fréchet distance lies in between the Hausdorff and the (weak) Fréchet distances:
    \[\delta_{\operatorname{H}}(P,Q) \leq \delta_{\operatorname{kF}}(P, Q) \leq \delta_{\operatorname{wF}}(P,Q) \leq \delta_{\operatorname{F}}(P,Q). \]
Also, the $k$-Fréchet distance decreases as $k$ increases: for $k=1$ it equals the weak Fréchet distance, whereas for $k$ sufficiently large, e.g., $k\geq n^2$, it equals the Hausdorff distance.

Figure~\ref{fig:comparison} illustrates this property.
The diagram on the left corresponds to a fixed $\eps_1$. We observe that there is one connected component in the free space $F_{\eps_1}(P,Q)$ that projects surjectively onto both parameter spaces. We therefore have $\eps_1\geq \delta_{\operatorname{wF}} \,(=\delta_{\operatorname{1F}})$.
The diagram in the middle depicts $F_{\eps_2}(P,Q)$ for a value $\eps_2$ slightly smaller than $\eps_1$. In that case two components cover the parameter spaces, which means $\eps_2\geq\delta_{\operatorname{2F}}$.
The free space $F_{\eps_3}(P,Q)$ shown on the right for an $\eps_3$ smaller than $\eps_2$ consists of three components and all three are necessary to cover the parameter spaces. Furthermore, reducing the value of $\eps_3$ even more would not split up the components into smaller subcomponents, but would just result in the set of all components not covering the parameter spaces any more. So we have $\eps_3 \geq \delta_{\operatorname{H}} = \delta_{\operatorname{3F}}$.

	\begin{figure}[ht]
		\centering
		\includegraphics[scale=0.6]{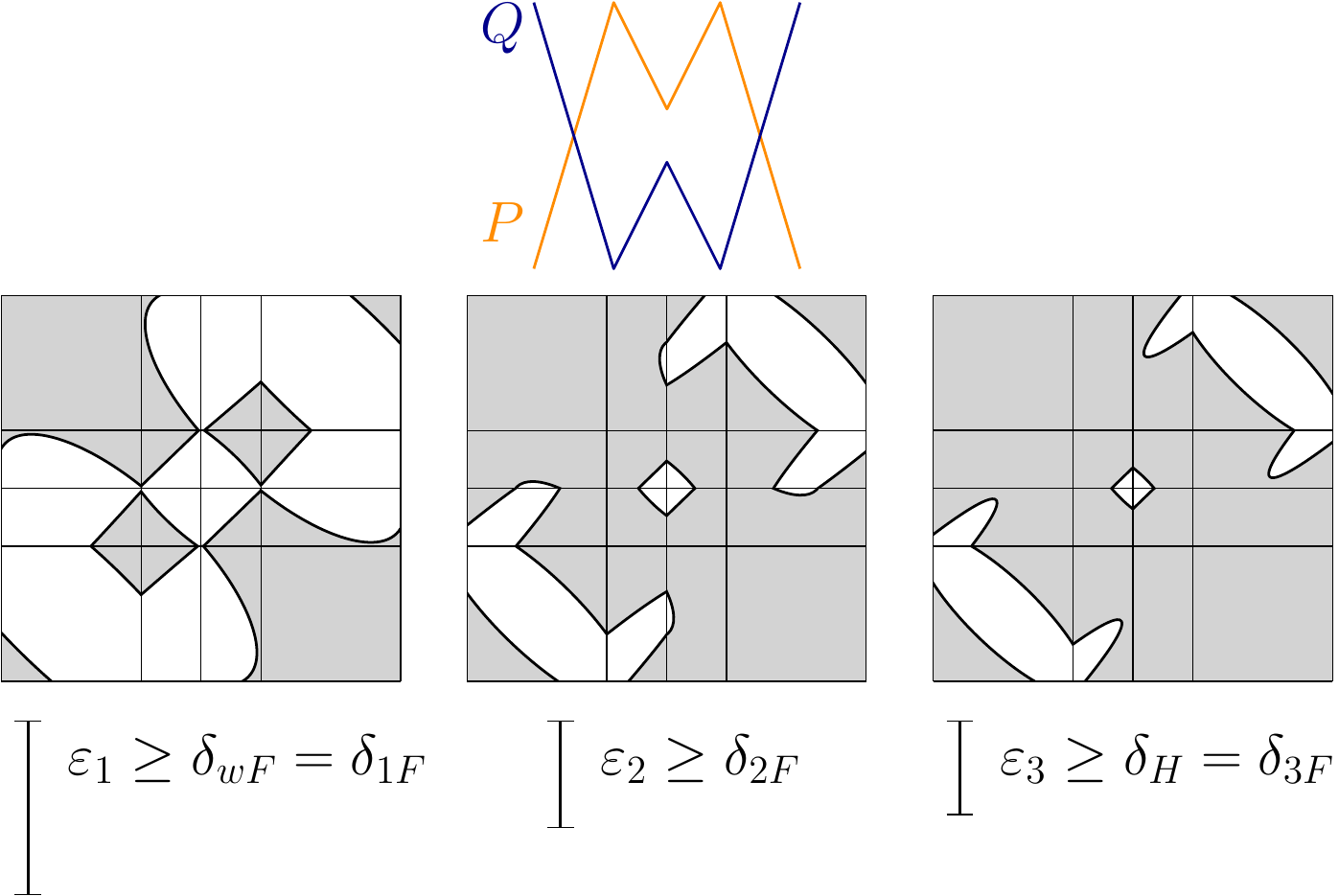}
		\caption{Comparison of weak Fréchet, $2$-Fréchet and Hausdorff distance of curves $P$ and $Q$.}\label{fig:comparison}
	\end{figure} 

\section{Hardness results}
\label{sec:cover-hard}
In this section we prove that the problem of deciding the $k$-Fréchet distance of two polygonal curves $P,Q$ is NP-complete. The optimization variant, i.e., minimizing the number of components $k$ is therefore NP-complete, too.

To give some intuition for the later proof we first present a reduction from the well-known 3-SAT problem to the problem of covering two sides of a rectangle by selecting a number of smaller rectangles, or boxes that are situated inside it. This problem (we call it the box problem, for short) mimics selecting the components in the free space to cover the parameter spaces but we do not ask to actually construct a free space, that is we do not ask to find curves that realize this specific free space.

Afterwards we do a 
reduction from rectilinear monotone planar 3-SAT \cite{sat} to prove hardness of the actual $k$-Fréchet distance problem: we first describe the gadgets we use, continue by constructing our actual curves and analyzing their complexity so that in the end we can give the reduction itself. 

\subsection{Gaining intuition for the free space: The box problem}
\label{sec:boxes}
We want to reduce from the following classical NP-complete satisfiability problem \cite{gareyjohnson}:

\vspace{\baselineskip}
\textbf{\textsc{3-SAT}:}\\
\textsc{Input:} a boolean formula with $n$ variables written as conjunction of $m$ clauses, where a clause is a disjunction of at most 3 literals;\\
\textsc{Output:} ``Yes'' if there exists an assignment of the variables such that the formula's output is true, ``No'' otherwise.\\
\vspace{0.5pt}

\noindent Here we want to use it to show hardness of the aforementioned problem:

\vspace{\baselineskip}
\textbf{\textsc{Box problem}:}\\
\textsc{Input:} a large rectangle $B$ with a set $A$ of smaller, interior-disjoint rectangles $b_i$ (all axis parallel) inside, a natural number $k$; \\
\textsc{Output:} ``Yes'' if there exists a selection of at most $k$ rectangles from $A$ such that their union surjectively projects onto the bottom and left boundary of the rectangle $B$, ``No'' otherwise.\\
\vspace{0.5pt}

Given any instance of a 3-SAT formula we want to build a bounding box $B$ containing a number of boxes $b_i$ such that we can find a covering selection of size $k$ if and only if there is an assignment for the formula that outputs true. A \emph{covering selection} of boxes is a subset of the $b_i$ that projects surjectively onto the bottom and left boundaries of $B$. For this we build boxes $b_i$ that correspond to the variables and any satisfying assignment of the variables can be directly ``translated'' into a covering selection of the $b_i$.
A complete construction is depicted in Figure~\ref{fig:np-boxes}, a detailed description of the coordinates can be found in Appendix~\ref{sec:appendix-box}.
We assume that no clause contains duplicates, i.e., no clause is of the form $v\lor v \lor w$. The duplicates can be deleted without changing the boolean function induced by the formula. Note that clauses of the form $v \lor \neg v \lor w$ are allowed.
Additionally we require that throughout the formula each literal appears at least once positive and at least once negated.
For each variable $v$ where this is not the case we add the clause $v\lor \neg v$ (colored dark green in Figure~\ref{fig:np-boxes}). These clauses are always fulfilled and therefore do not change the output of our boolean formula. We add at most $n$ clauses in this way, which means that the size of the formula only changes polynomially in the input size.

	\begin{figure}[ht]
		\centering
		\includegraphics[scale=0.6]{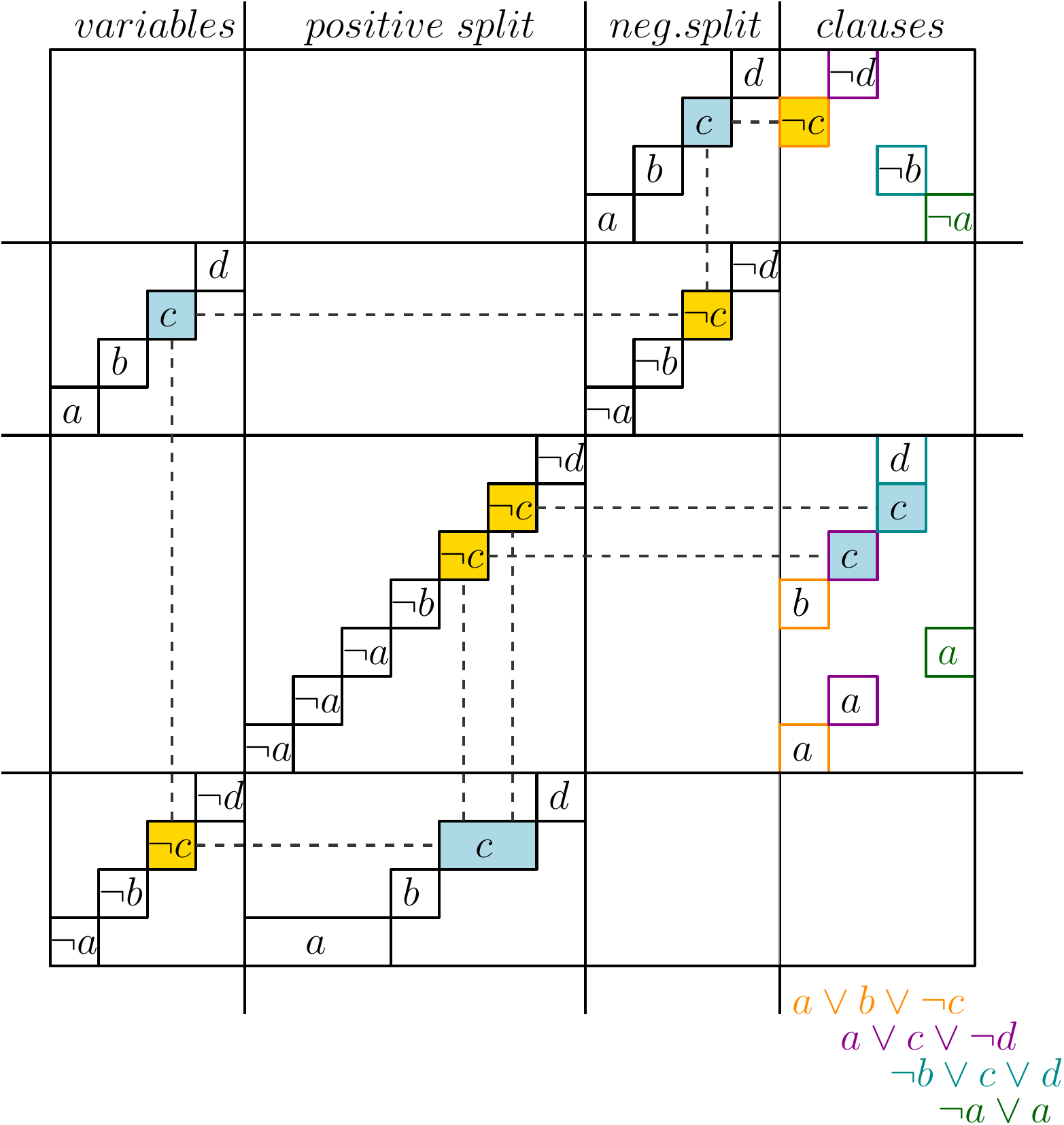}
		\caption{Construction of the box problem instance and propagation of assignment.}\label{fig:np-boxes}
	\end{figure}

In Figure~\ref{fig:np-boxes} we constructed a box problem instance from the formula $(a\lor b \lor \neg c) \land (a \lor c \lor \neg d)\land (\neg b \lor c \lor d)$. Since $\neg a$ does not occur in any clause we added the clause $(\neg a \lor a)$. 

The construction briefly works as follows: for each variable $v$ we add two boxes in the variable columns, one labeled $v$ and the other one $\neg v$. The number of occurrences of each literal is treated in the split columns: we add wider boxes for literals that occur multiple times and as many unit width boxes of the negated literal above. Finally, we add boxes corresponding to occurring literals in the rightmost columns, where a unit width column corresponds to a clause, respectively.
Note that each unit row contains exactly two boxes of opposite labels.
We set $k$ to be equal to half the number of boxes, so in order to cover the vertical boundary we have to choose one box per row.
The bounding box of the small rectangles is called $B$.

Now, to cover the variable columns we need to chose an assignment: either we select a variable or its negated version. As displayed for variable $c$, the choice made for the first column determines the selection of boxes for the split gadgets and the clauses. The colored boxes correspond to choosing $c=\mbox{\it true}$ (blue) or $c=\mbox{\it false}$ (orange). Note that both choices imply selecting the same number of boxes. In order to cover the clause columns we need to select at least one literal of each clause.
The proof of the following theorem can also be found in Appendix~\ref{sec:appendix-box}.

\begin{theorem}\label{thm:box}
The box problem is NP-complete.
\end{theorem}

We can interpret the box problem as the problem of finding a selection of components in the free space that cover the parameter spaces. The small boxes can be seen as bounding boxes of actual components (for the projection there is no difference) and the bottom and left boundary of the large box $B$ correspond to the parameter spaces. 
This hardness proof, especially the construction of the boxes, provides us with the key ideas to prove hardness of the $k$-Fréchet distance. Next, we construct actual curves where certain intervals on the parameter spaces of the free space diagram each have two components that could cover them. As with the box problem, the choice we make for one of those intervals determines the choices for other intervals as we still need to ensure that the selection size is minimal in the end. The propagation of choices works in the same manner for the box problem as for the $k$-Fréchet distance problem.

\subsection{Construction}
\label{sec:cover-hard-construction}

We use the following variant of the 3-SAT problem in this subsection.

\vspace{\baselineskip}
\textbf{\textsc{Rectilinear monotone planar 3-SAT}:}\\
\textsc{Input:} a 3-SAT formula with only all positive or all negated variables in each clause, embedded as a graph where all edges are rectilinear and non-crossing; variables are drawn as vertices on a horizontal line, positive clauses are vertices drawn above this line and negative clauses are drawn below;\\
\textsc{Output:} ``Yes'' if there exists a satisfying assignment for the variables, ``No'' otherwise.\\
\vspace{0.5pt}

Note that we assume that each variable appears in at least one positive and one negative clause. Otherwise we could simply define the occurring literal to be true (or false, respectively) and omit the clauses the literal appears in because they would be fulfilled regardless of the assignment of the other participating literals.

	\begin{figure}[ht]
		\centering
		\includegraphics[scale=0.6]{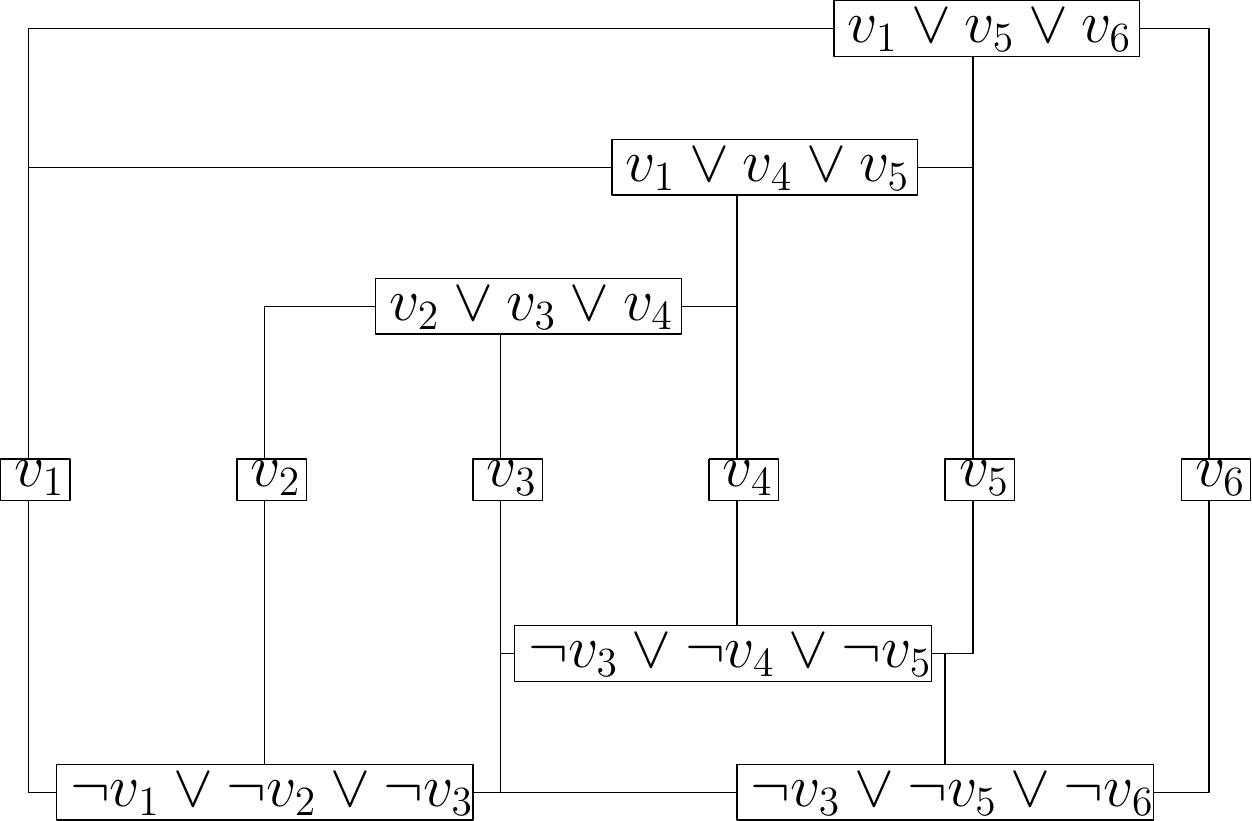}
		\caption{Instance of rectilinear monotone planar 3-SAT.}\label{fig:input}
	\end{figure}

As apparent, we can draw any such graph on a grid, which is useful when constructing our curves and analyzing their complexity.
Since rectilinear monotone planar 3-SAT is NP-complete \cite{sat}, we prove hardness of the cover distance problem by reduction from it.

\vspace{\baselineskip}
\textbf{\textsc{$k$-Fréchet distance problem}:}\\
\textsc{Input:} Two polygonal curves $P$ and $Q$, a distance $\varepsilon$ and a natural number $k$;\\
\textsc{Output:} ``Yes'' if there exists a selection of at most $k$ components in the free space diagram $F_\varepsilon$ such that their union projects surjectively onto both parameter spaces, ``No'' otherwise.\\
\vspace{0.5pt}

Our goal is to construct curves that mimic any input instance of a rectilinear monotone planar 3-SAT graph and show that in the free space resulting from these curves we can find a covering selection of size $k$ if and only if there exists a satisfying assignment for the formula. The full construction can be found in Appendix~\ref{sec:appendix-hardness}.

Overall we create wire and clause gadgets to represent variables and clauses,
where wires correspond to edges of the input graph.
They are connected as the given embedding of the 3SAT instance.
Wire gadgets allow a boolean choice that is propagated consistently throughout the wire.
Clause gadgets test whether at least one incoming wire caries an appropriate choice.

Figure~\ref{fig:wire1} shows a wire gadget. Both the yellow and the blue curves run along the sides (the vertical parts of the curves, which we call \emph{base curves}) and form \emph{spikes}. The sides are uninteresting for the analysis because the segments forming it can only be covered by larger components that are always part of any covering selection.
The value $\eps$ is chosen such that two adjacent spikes are just within distance $\eps$.
It follows that the spikes induce components in the free space diagram that are similar to the boxes of Subsection~\ref{sec:boxes}.
We say that a spike $s$ is \emph{covered} by an adjacent spike $t$ of the other curve if the component of the free space diagram that covers the two intervals induced by these spikes is chosen for the covering selection. After the construction we choose $k$ such that each blue spike in any gadget can only be covered by one single adjacent yellow spike. The choice for one blue spike 
must be consistent along
the wire and encodes the assignment of the corresponding variable.

\begin{figure}[ht]
		\centering
		\includegraphics[scale=0.7]{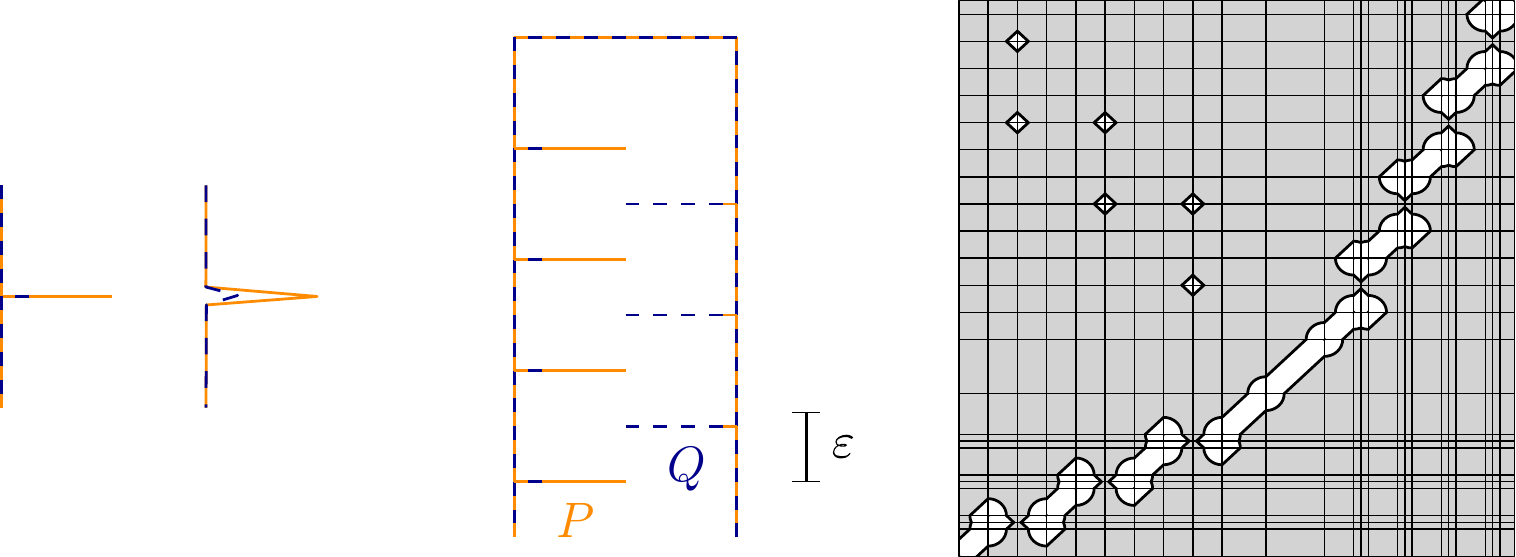}
		\caption{(Left) A spike and a small perturbation of it. (Right) The wire gadget and its corresponding free space diagram.
		Note that we connected the curves to give a small example, but the horizontal segment on top is not part of the gadget itself.}\label{fig:wire1}
	\end{figure}

Of course, we need a number of other gadgets, too. As mentioned, the wires correspond to edges in the rectilinear monotone planar 3-SAT instance. To draw them coherently we need to make sure we can make 90° turns (so called \emph{bends}) and do T-crossings, i.e., \emph{split} a wire into two. Last but not least we need to build a \emph{clause gadget} where three wires connect. In Figure~\ref{fig:gadgets}, we show a bend and a clause gadget, connected through wires. The only basic gadget not shown in Figure~\ref{fig:gadgets} is the split, which looks similar to the clause.

\begin{figure}[ht]
		\centering
		\includegraphics[scale=0.7]{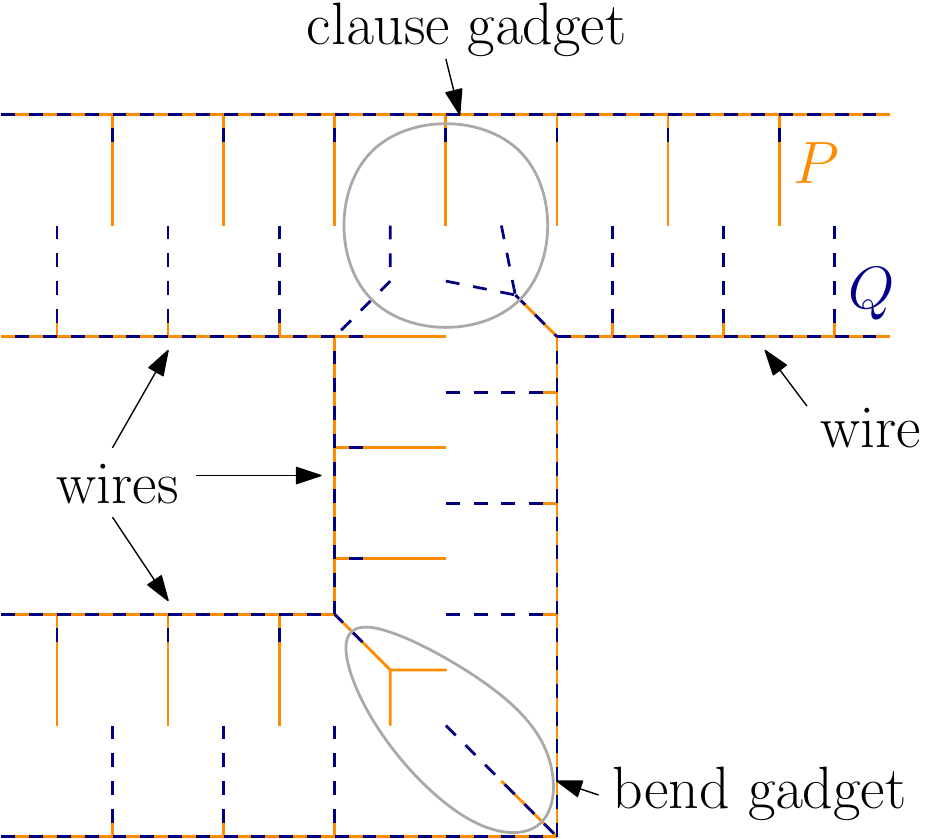}
		\caption{A small example part of the constructed curves consisting of wires, a bend and a clause gadget.}
		\label{fig:gadgets}
	\end{figure}

Additionally, we need to make sure that both curves are connected. In order to do so, we establish a number of other gadgets: first of all, there is a \emph{connection gadget} that enables us to connect the two base curves of $P$ and $Q$, respectively. The resulting curves are closed; to prove hardness for non-closed curves we apply the \emph{scissor gadget}. Finally, we sometimes need to change which of the curves has  spikes on a specific side to draw the other gadgets consistently, so we also built a \emph{color gadget} to ``switch'' the color pattern of the spikes.

At last we want to connect all gadgets such that the resulting curves follow the embedding of the input graph $G$. Recall that the 
input is a grid embedding.
We first scale the grid by a factor of $2^{10}$ to place all gadgets consistently. Note, that we have to deal with 2-clauses and take into account that our split gadget is directed, so we need workarounds for cases where the split has to face in other directions. Afterwards we are able to draw the curves' vertices on grid points only. 
To give an order in which the gadgets are traversed, we consider the input graph $G$. We want to traverse all edge of $G$ twice, once per inner, once per outer base curve. To do so, we have to ``walk around'' each face of $G$. To switch between faces we use connection gadgets. We obtain a traversal order of the faces by computing a minimum spanning tree of the dual graph. 
A traversal of the tree visits each node at least once (see Appendix~\ref{sec:appendix-building} for a detailed description). 

Finally, it remains to prove that our construction works in the sense that the curves have $k$-Fréchet distance $\eps$ if and only if the specific 3-SAT instance is satisfiable. 

\subsection{Analysis}
\label{sec:cover-hard-analysis}
First, we note that the complexity of our curves is polynomial in the size of our input instance: the numbers of variables and clauses, but also the number of splits and the length of the edges determine the number of spikes and therefore also the number of components in the free space diagram. A spike induces either two or three components, depending on whether it is part of a specific gadget, i.e., a clause, or not.
In addition, the gadgets induce a number of components, called \emph{clutter}, that are always part of a covering selection.
Some gadgets also induce a constant number of unnecessary components that are never chosen.

Our goal is to cover the parameter spaces with $k$ components. We definitely need to select all clutter components and we need to cover all spikes, therefore we need to select (at least) one component per spike.
We set $k$ to be the number of clutter components plus the number of blue spikes (spikes of $Q$).
It follows that each blue spike can only be covered once, which ensures that choices are propagated.

\begin{theorem}\label{thm:np}
It is NP-complete to decide whether $\delta_{\operatorname{kF}}(P,Q) \leq \eps$ for given polygonal curves $P$ and $Q$, integer $k$, and $\eps>0$ where $\delta_{\operatorname{kF}}$ denotes the $k$-Fréchet distance.
\end{theorem}


The full proof is given in Appendix ~\ref{sec:appendix-analysis}. For the reduction constructed above, the following holds:
given a satisfying assignment for the input formula, we know which components to select: apart from all clutter components,
we have to decide how to cover the blue spikes.
This choice is implied by the assignment and propagated throughout the gadgets.


Given a selection of components, we need to backtrack our choices throughout the wires and other gadgets to determine how the blue spikes are covered. Depending on this choice, we know whether the corresponding variable has to be set to true or to false. Thus we derive our assignment for the 3-SAT formula and complete the proof of NP-hardness.

We can test in polynomial time whether the union of a selection of components covers the parameter spaces. Thus the problem of deciding the $k$-Fréchet distance lies in NP. 

\section{Algorithms}
\label{sec:cover-algo}
In this chapter, we begin by presenting a preprocessing algorithm, which is applicable to the following algorithmic approaches: 
First, we can find a covering selection of at most size $k$ in exponential time (Section~\ref{sec:xp}).
Then, we describe an FPT-algorithm for finding an optimal covering selection in Section~\ref{sec:fpt}.
At last, we compute a selection of at most twice the optimal size, i.e., we approximate $k$ in Section~\ref{sec:approx}.
Note that the XP-algorithm as well as the FPT-algorithm are designed to solve the decision problem, but we could also optimize $k$ by repeating the decision problem solving algorithm for different values of $k$. We could perform a parametric search on the reasonable values for $k$, similar to the algorithm for the Fréchet distance for polygonal curves by Alt and Godau~\cite{altgodau}.

\subsection{Preprocessing}
First, we observe two preprocessing strategies, which can be applied before entering either one of the later discussed algorithms.
In any case we start by computing the free space diagram, which takes quadratic time. In the free space diagram it is easy to identify all \emph{necessary} components: any component that covers an interval of one of the parameter spaces uniquely (i.e., there is no other component covering the exact same interval) is necessarily chosen for an output selection. Such components can be found in $\mathcal{O}(n\log n)$ time using a scan. Furthermore, it is possible to rule out all \emph{redundant} components, which only cover intervals on both parameter spaces that are already covered by (at least) one other component. To be precise a component is called redundant if and only if it is completely contained in the bounding box of a different component (but there could be more than one such component with a sufficiently large bounding box). This case can also be detected via scans.
Thus our preprocessing needs quadratic time. However, it does not improve the size of the input (being the complexity of the free space) nor the resulting runtime of any of the presented algorithms asymptotically.

\subsection{Exponential-time brute force approach}\label{sec:xp}

As a first approach, we present an XP-algorithm.

\begin{remark}
The $k$-Fréchet distance can be decided in $\mathcal{O}(k\cdot n^{2k})$ time for constant $k$.
\end{remark}

The brute force approach simply checks for all selections of $k$ components of the free space whether their joint projections cover both parameter spaces surjectively. That means we have to check at most $\binom{n^2}{k}$ possible combinations of components resulting in a runtime of $\mathcal{O}(k \cdot n^{2k})$ for fixed $k$, which is of course only feasible for very small $k$.
Therefore we can compute the answer to the decision problem for the cover distance with $k=2$ in $\mathcal{O}(n^4)$.
Since $\binom{m}{k} \leq 2^m$ holds for any $m > k$, our runtime is upper-bounded by $\mathcal{O}(n\cdot 2^{n^2})$ for general~$k$.

\subsection{Fixed-parameter tractability}\label{sec:fpt}

Next, we present an algorithm for deciding whether $\delta_{\operatorname{kF}}(P,Q) \leq \eps$ for given $\eps$ and $k$. The runtime of our algorithm is polynomial in the complexity of our curves $P$ and $Q$, but exponential in the two parameters $k$ (the selection size) and $z$ (the neighborhood complexity). 

We define $z$, the \emph{neighborhood complexity} of the curves as the maximum number of segments of one curve that intersect with the $\eps$-neighborhood of any point of the other curve. That is, in the free space diagram we get that each horizontal and each vertical line intersects at most $z$ components.

The idea of the algorithm is the following: we build two directed bounded search trees (as described in Chapter~3 of \cite{boundedsearchtree}) to create selections of components of size at most $k$. Each search tree represents the projection of the free space onto one parameter space, see Figure~\ref{fig:fpt} below.
A node corresponds to a component in the free space (or rather the interval on the respective parameter space it covers) and a path encodes a selection that covers the interval $[ l,r]$ of the respective parameter space. By $l$ we denote the left boundary point of the interval corresponding to the topmost node of the path (e.g.\ the root) and $r$ is the right boundary point of the interval of the bottommost node (e.g.\ a leaf). 
We call a selection or a path \emph{feasible} if the union of the (at most $k$) components it contains/encodes/covers the respective parameter space. From the first tree, $T_P$, we are able to extract all feasible selections which cover the parameter space corresponding to curve $P$, feasible selections of the second tree, $T_Q$, cover the other parameter space.
In the end we want to compare and/or combine a feasible selection of $T_P$ with a feasible selection of $T_Q$ to get a selection $S$ that contains no more than $k$ components, so that its union projects surjectively onto both parameter spaces.

\begin{figure}[ht]
		\centering
		\includegraphics[scale=0.7]{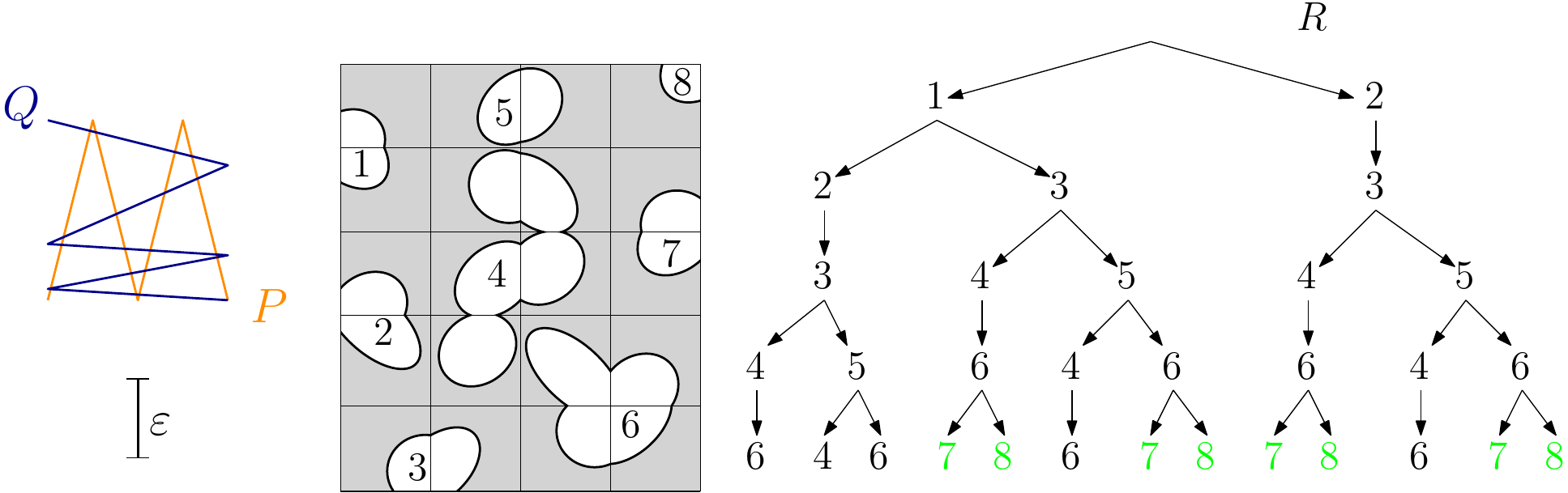}
		\caption{Curves $P,Q$, their free space diagram and the resulting bounded search tree $T_P$. Leaves of feasible paths are marked in green.}\label{fig:fpt}
	\end{figure}

More formally, we build two trees of depth $k$ and branching factor $z$. Consider the tree $T_P$. The root is labelled by the left boundary point of the parameter space of $P$ (we assume w.l.o.g. that the bottom boundary of the free space diagram corresponds to $P$). Now we use a sweep line initialized at the left boundary of the free space diagram. We assign a node in the tree to all components intersecting the sweep line, i.e., the root has as many children as there are components touching the left boundary of the free space diagram. The sweep line moves to the right. Whenever the sweep line is tangent to a component, one of two cases occur: if it touches the leftmost point of a component it becomes \emph{active}, i.e., the sweep line continues to intersect this component when moving further to the right; if the line touches the rightmost point of the component, it becomes \emph{inactive} (so the sweep line just stops to intersect it). In the first case nothing immediate happens to the tree, in the latter case, if the tangent component already has a node in the tree, we insert new nodes: each node corresponding to the tangent component gets assigned as many children as there are other currently active components. By definition a node can never have more than $z$ children. Note that some (small) components may not get assigned any nodes in one tree.
Also, with every node we store its depth (the root has depth 0)  and stop assigning children at depth $k$ - or as soon as a component touches the right boundary of the free space diagram. If a leaf $v_l$ corresponds to a component touching the right boundary, the path from the root to $v_l$ encodes a feasible selection of components for $T_P$. Other selections are called \emph{non-feasible}.
 The second tree $T_Q$ is built analogously by sweeping from bottom to top.

We store the feasible selections obtained from $T_P$ and $T_Q$ in sorted lists $L^T_P$ and $L^T_Q$. For each pair of selections $S_{P,i}$, ${S_{Q,j}}$, where $1\leq i,j \leq z^k$, we test whether $\vert S_{P,i} \cup S_{Q,j}\vert \leq k$ and output this union if the answer is positive.

\begin{theorem}\label{fpt_thm}
The algorithm described above returns a selection $S$ of $k$ components in the free space that surjectively projects onto both parameter spaces if and only if such a selection exists. Therefore it solves the decision problem for the $k$-Fréchet distance in time $\mathcal{O}(nz + kz^{2k})$.
\end{theorem}

\begin{proof}
Our algorithm treats all possible selections of size at most $k$ per parameter space and combines all these, hence it necessarily finds a valid solution if and only if one exists.

For the first step, we compute the free space, which takes $\mathcal{O}(n\cdot z)$ time. Building the trees takes $\mathcal{O}(z^ k)$ time since we are limited to depth $k$ and insert at most $z$ children per node. Note that any operation in the free space diagram, such as detecting components that cover a boundary or projecting them onto the boundaries to find the intervals they cover, can be done in $\mathcal{O}(n)$ time.
We obtain at most $z^k$ selections per search tree. Each selection is stored as the sorted set $S_{P,i}$, respectively $S_{Q,i}$, by encoding each component as an integer. We then sort both lists of selections lexicographically. During the sorting we might detect duplicates, which we discard immediately. We then compare each selection of the first list $L^T_P$ to each selection of $L^T_Q$, taking time $k$ per comparison. For any selection smaller than $k$ we can test whether its union with a selection of the other list is still a solution, i.e., whether the unified selection does not have more than $k$ integers. All in all we have a runtime of $\mathcal{O}(z n+z^k+z^k \cdot k\log k+k(z^k)^2)=\mathcal{O}(zn + kz^{2k})$.
\end{proof}

\subsection{Approximation}\label{sec:approx}

Finally, we discuss how to approximate the size $k$ of an optimal solution.
The main idea of our algorithm is to find minimal covering selections for each parameter space individually and combine those selections into an overall solution in the end. We can find both selections covering only a single parameter space by applying a greedy technique.

Given the free space diagram, we first project all components onto the parameter spaces. We get two sets of intervals, one covering the first parameter space (we store these intervalls in the list $L_P$, see Figure~\ref{fig:approx}) and one for the second parameter space (stored in $L_Q$). So one component projects onto two intervals, one on each parameter space (and therefore one per list). We store the information on which two intervals stem from the same component accordingly.

\begin{figure}[ht]
		\centering
		\includegraphics[scale=0.7]{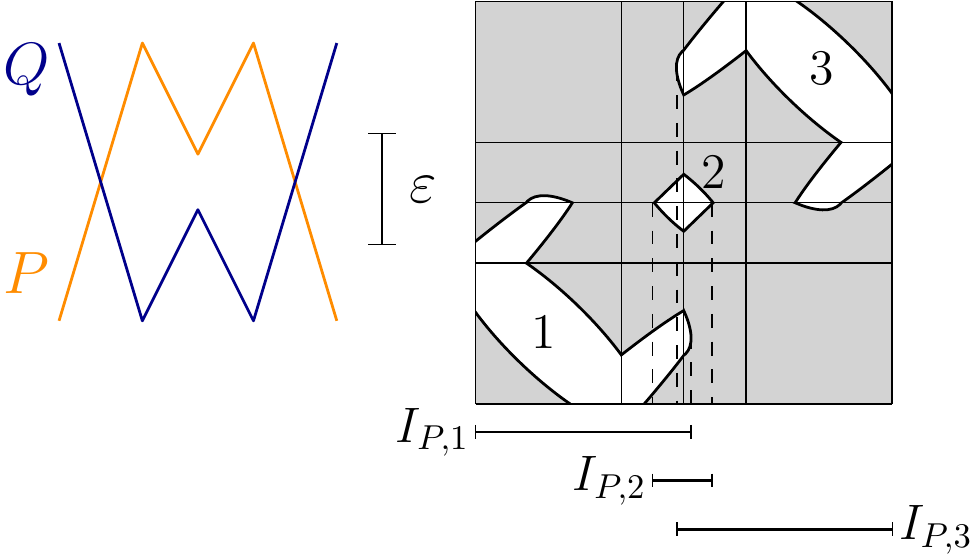}
		\caption{The projection onto the first parameter space and the resulting elements of $L_P$.}\label{fig:approx}
	\end{figure}

Now we simply want to select a minimum number of these intervals whose union equals the unit interval, i.e., the parameter space. We deal with each parameter space on its own in the following way: we sort the lists $L_P$ (and $L_Q$) by left endpoint.
Now, per list, starting at 0, we make a greedy choice and select the interval (among the intervals starting at 0) with the rightmost endpoint, say $r_1$. 
Here we recurse, i.e., we take $r_1$ as new start point and again search among the intervals covering $r_1$ (i.e., intervals starting at or to the left of $r_1$) for the one with the rightmost endpoint. 
As soon as we select an interval with 1 as endpoint we have found a minimal selection covering the respective parameter space. 
To see that our greedy strategy is optimal, observe that the algorithm proceeds from left to right maintaining the following invariant: 
at any time we selected a minimum number of intervals to cover the parameter space from the left boundary to the current position.

As output we have two selections of intervals, $S_P$ and $S_Q$. The intervals correspond to components. We build the union of both lists, taking into account that an interval in $S_P$ may belong to the same component as an interval of $S_Q$, and output the selection of components $S$ that contributed at least one of the chosen intervals.

The worst case that might occur is the following: all of the intervals we selected during the greedy procedures correspond to different components in the free space, so that the union of our selections is of size $\vert S_P\vert+\vert S_Q\vert$. A different selection of size $\vert \hat{S}\vert =\max(\vert S_P\vert ,\vert S_Q\vert)$ might cover both parameter spaces but is not detected by the greedy scan. This happens due to the fact that the greedy selection only detects one optimal solution for either parameter space, not all of them, and does not take into account that intervals possibly correspond to the same component.

Finally, we consider the runtime: computing the free space takes quadratic time. Sorting the lists adds another logarithmic factor while the greedy selection routine takes linear time in the number of intervals.
Hence we get an overall runtime of $\mathcal{O}(n^2\log n)$.

\begin{theorem}
The algorithm described above runs in $\mathcal{O}(n^2\log n)$ time and finds a selection of components that covers both parameter spaces if and only if one exists. A found selection contains at worst twice the minimum number of components needed.
\end{theorem}

We conjecture that the approximation factor $2$ is probably not tight: one can find a free space diagram for which the algorithm's output is twice the size of the optimal solution, but it seems hard to construct curves that realize such a diagram.
Showing that it is or that a better approximation factor holds remains open.

\section{Conclusion}
\label{sec:cover-conclusion}
We presented a novel variant of the Fréchet distance that enables us to detect similarities between objects that resemble each other only piecewise. We ask for a number $k$ of (possibly overlapping) subcurves of our input curves, such that we can find pairs of subcurves that have small weak Fréchet distance. Therefore the $k$-Fréchet distance provides a transition between weak Fréchet ($k=1$) and Hausdorff distance ($k$ sufficiently large, e.g., $k = n^2$).

Unfortunately, deciding the $k$-Fréchet distance of two polygonal curves proved to be NP-hard (naturally, minimizing the number of subcurves  is as well). However, we were able to tackle the computational challenge from different angles: we gave an XP-algorithm depending on $k$, presented an FPT-algorithm using two parameters
and even found a polynomial time algorithm that gives a 2-approximation on the number of subcurves.
For the approximation algorithm, it remains open to prove tightness of our approximation factor. 

As mentioned in the introduction, there is also a second variant of defining the $k$-Fréchet distance we call the ``cut version'': instead of allowing overlapping subcurves, we cut the curves into several pieces and search for a matching between the pieces of the two curves. The hardness proof of Buchin and Ryvkin~\cite{eurocg} holds for this variant, i.e., it is also NP-hard. However, the algorithmic approaches only work for the variant we discuss in this paper (we call it the ``cover variant''). Finding algorithmic approaches and examining hardness for the cut version of the $k$-Fréchet distance is work in progress. We conjecture that the decision problem for the cut variant is $\exists\mathbb{R}$-complete.

\bibliography{kFrechetArXiv}

\newpage
\appendix
\label{sec:appendix}
\section{The box problem} \label{sec:appendix-box}

Here we give the detailed construction of our box problem instance derived from any 3-SAT formula:
Let $V=\{v_1, \dots, v_n\}$ be the set of variables and let $C=\{c_1,\dots,c_m\}$ be the set of clauses.
For each variable $v_i$, let $a^+_i$ (respectively $a^-_i$) be the number of clauses in which $v_i$ appears positive (respectively negated) and let $ \{c^+_{i, 1}, c^+_{i, 2}, \dots, c^+_{i, a^+_i}\}$ (respectively $\{c^-_{i, 1}, \dots, c^-_{i, a^-_i}\}$) be the set of clauses in which $v_i$ appears positive (respectively negated).
Additionally we define the sums $s^+_i=\sum_{k=1}^i a^+_k$ and $s^-_i=\sum_{k=1}^i  a^-_k$. 

	\begin{figure}[ht]
		\centering
		\includegraphics[scale=0.6]{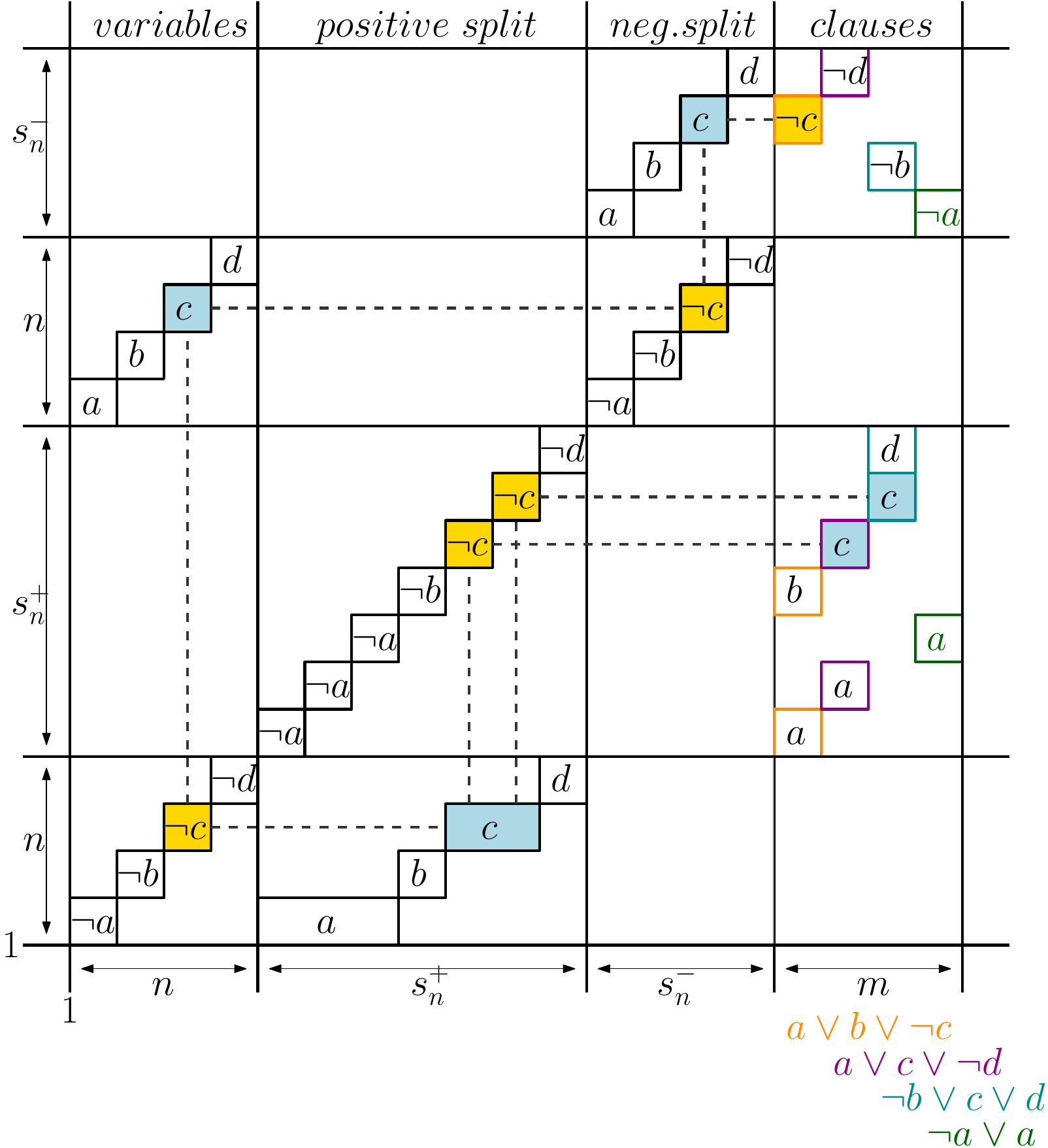}
		\caption{Construction of the box problem instance and propagation of assignment.}\label{fig:np-boxes-appendix}
	\end{figure}

In the following we describe the placement of boxes as depicted in Figure~\ref{fig:np-boxes-appendix}. The number of rows and columns needed for the different gadgets is indicated in the figure. A box $(x, y, w, \ell)$ designates the axis-aligned rectangle with unit height and width $w$ whose bottom left corner has coordinates $(x, y)\in \R^2$ with label $\ell$. The labels are later used in the proof of correctness.

\textbf{Variable gadget.}
For each variable $v_i$, we place two boxes $(i, i, 1, \neg v_i)$ and $(i, i + n + s^+_n, 1, v_i)$, and no other boxes are placed over the interval $(i, i+1)$ of the bottom boundary. That way, in order to cover said interval, at least one of those two boxes has to be chosen.

\textbf{Split gadget.}
We build the splits used for the positive occurrences of the variables first.
For each variable $v_i$, we place the box $(1+n+s^+_{i-1}, i, a^+_i, v_i)$ and the boxes $(n+s^+_{i-1}+j, n+s^+_{i-1}+j, 1, \neg v_i)$, for $j\in \{1, \dots, a^+_i\}$.
For negated occurences of $v_i\in V$ we place the box $(1+n+s^+_n+s^-_{i-1}, n+s^+_n+i, a^-_i, \neg v_i)$ and the boxes $(n+s^+_n+s^-_{i-1}+j, 2n+s^+_n+s^-_{i-1}+j, 1, v_i)$, for $j\in \{1, \dots, a^-_i\}$.

\textbf{Clause gadget.}
We assign to each clause $c_i$ the unit interval on the bottom boundary of $B$ starting at $I(c_i) = n+s^+_n+s^-_n+i$.
For each literal of a clause $c_i$ we place a box labelled with the respective literal above the unit interval $[I(c_i),I(c_i)+1]$.
To be precise, for each $v_i\in V$ we place the boxes $(I(c_h), n+s^+_{i-1}+j, 1, v_i)$, for $j\in \{1, \dots, a^+_i\}$ and $h \in \{1, \dots m\}$ where $c_h=c^+_{i,j}$, and $(I(c_h), 2n+s^+_n+s^-_{i-1}+j, 1, \neg v_i)$, for $j\in \{1, \dots, a^-_i\}$,  $h \in \{1, \dots m\}$ where this time $c_h=c^-_{i,j}$.

Overall this means that we have $4n + 2(m_1+2m_2+3m_3)$ boxes, where $m_i$ is the number of clauses with $i$ variables (and therefore $m_1+m_2+m_3=m$).
Each unit interval $(i, i+1)$ with $i\in \{1, \dots, 2n+s^+_n+s^-_n\}$ on the left boundary of $B$ can be covered by exactly two different boxes. The same holds for every unit interval $(i, i+1)$ with $i\in \{1, \dots, n+s^+_n+s^-_n\}$ on the bottom boundary.
Note that for all these unit intervals, one of the boxes is labeled with a variable and the other one is labeled with the negated version of that variable, i.e., one box is labeled $v$ and the other one $\neg v$.
Each Interval $I(c)$ on the bottom boundary can be covered by as many boxes as the clause $c$ contains literals. The labels of these boxes correspond to the variables contained within this clause.
We set the bounding box $B$ as the axis-aligned rectangle spanned by the points $(1,1)$ and $(1+n+s^+_n+s^-_n+m, 1+2n+s^+_n+s^-_n)$ and we set $k=2n + m_1+2m_2+3m_3$ so only half the boxes can  be chosen.
For a given boolean formula, the set of boxes defined above can be determined in polynomial time.

Next we prove Theorem~\ref{thm:box}, i.e., the NP-completeness of the box problem.

\begin{proof}
First we prove that the box problem as constructed above has a solution if and only if the input 3-SAT formula has a variable assignment such that it evaluates to true.

\textbf{'$\Leftarrow$'}
Let $f:V\rightarrow \{true, false\}$ be an assignment of the variables that satisfies the 3-SAT formula.
We set $S= \{\text{boxes }(x, y, w, v) \mid f(v)=true\}\cup \{\text{boxes }(x, y, w, \neg v) \mid f(v)=false\}$.
The set $S$ projects surjectively onto the bottom and left boundary of the bounding box $B$ because each unit interval on the left boundary is covered by exactly one box. For most of the bottom boundary we also have that each interval is uniquely covered, but for the clauses columns we allow that more than one box per unit interval is chosen (i.e., more than one corresponding literal is set to true).

\textbf{'$\Rightarrow$'}
Let $S$ be a minimal set of boxes that covers the boundaries of the bounding box $B$ with $|S|=k$.
This means that each unit interval on the left boundary of $B$ is covered by exactly one box. Due to the position of the boxes, this means that for each variable $v$ either all boxes labeled $v$ or all boxes labeled $\neg v$ have been chosen. This induces an assignment of the variable $v$, i.e., $v$ is set to true if the boxes labeled $v$ have been chosen and else $v$ is set to false.
Note that the selection $S$ covers the box $B$. Therefore, for each clause $c$ one of the boxes that can cover $I(c)$ is an element of the selection $S$. It follows that the assignment of variables induced by $S$ fulfills the formula.


Above we showed the NP-hardness of the box problem. The box problem is NP-complete since for a given subset $S$ of boxes one can test if the bounding box $B$ is covered by simply marking the covered intervals, which can be done in polynomial time.

\end{proof}

\section{Hardness of the k-Fréchet problem}\label{sec:appendix-hardness}

In this section we give the detailed description of the gadgets needed for the reduction from rectilinear monotone planar 3-SAT to the $k$-Fréchet distance problem and present the formal proof of NP-hardness.

\subsection{Gadgets}\label{sec:appendix-gadgets}

In this subsection we first describe the basic gadgets we need for our reduction. The curves are intricate, which is why we also need some other gadgets to be able to draw them correctly. In Subsection~\ref{sec:appendix-building} we describe how to combine the gadgets to construct the actual curves $P$ and $Q$ corresponding to the input instance. 

\subsubsection{Main gadgets}

Assigning a variable to be true or false corresponds to choosing between one of two options in the free space diagram, just as choosing between two boxes per row for the box problem. Therefore we do not build specific variable gadgets but encode the variables in the wires, which work as ``edges'', i.e., connections between other gadgets. Wires should produce similar ``staircase'' structures as in Figure~\ref{fig:np-boxes} (or Figure~\ref{fig:np-boxes-appendix}) in order to provide a choice between two components per uncovered interval of either parameter space.

	\begin{figure}[ht]
		\centering
		\includegraphics[scale=0.8]{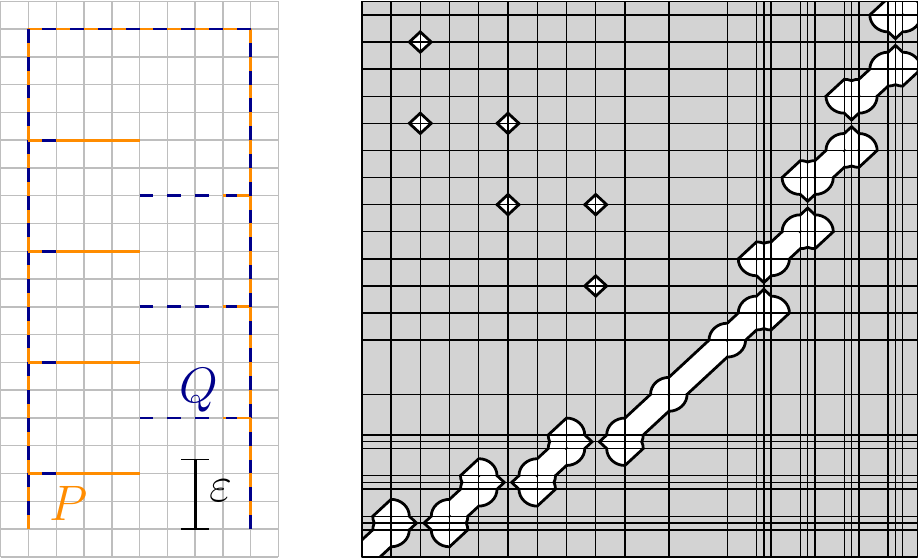}
		\caption{Wire gadget drawn on a (coarse) grid with its resulting free space diagram.}\label{fig:wire-appendix}
	\end{figure}

\textbf{Wire gadget.} Figure~\ref{fig:wire-appendix} displays a wire gadget: we build two \emph{base curves} that are intersected with \emph{spikes} of different lengths: each curve consists of two parallel base curves (in this figure they are drawn vertically) with shorter spikes on one and longer spikes on the other side (horizontal segments that are pointing inwards). Short spikes of one curve lie on top of the longer spikes of the other one. In the following the term ``spike'' refers to long spikes. We also show an underlying grid on which we draw our gadget, but note that we only drew every fourth line to enhance readability. The distance between two spike tips of the same color is therefore 16 times the unit distance.
Two tips of spikes that are within $\eps$-distance (i.e., within a distance of 10 times the unit length) are called \emph{adjacent}.
The spikes generate the small components forming a \emph{staircase} while the segments on the base curves generate components that we definitely have to take. These larger components (that are the only ones covering a certain interval on either parameter space) are called \emph{necessary clutter}, or, for short, clutter, throughout the rest of this paper.
Sometimes smaller components are generated that are covering subsets of the intervals covered by clutter components, so called \emph{unnecessary} components. They are never chosen for a minimal covering selection. 

Between the clutter components there are gaps, namely the intervals induced by the spikes, that can be covered by one of two components induced by the two adjacent spikes of the the opposite color.
Remember that solving an instance of the $k$-Fréchet distance problem means selecting a set of components from the free space diagram, which covers both parameter spaces. When we choose a component that covers the interval induced by a spike we say that the spike is \emph{covered} by this component, respectively by the adjacent spike inducing said component. 

At the end of the construction the value $k$ is chosen such that each blue spike (curve $Q$) can only be covered by one single yellow spike (curve $P$).
The spikes of $Q$ determine the variable assignment: we define a \emph{central spike} in each wire that encodes a variable.
As evident from Figure~\ref{fig:input}, we have one vertical edge per variable of the input instance, which connects the positive clause (or clauses) the variable takes part in to the negative one(s). As we construct the wires to follow the edges for each variable $v$, we define the central spike of $v$ to be a blue spike in this vertical part.
If the central spike is covered by its upper yellow neighbor, the variable corresponding to this wire is set to true; if the central spike is covered by its lower neighbor, the variable is set to false.
The choice made for the central spike is propagated throughout the rest of its wire.
The fact that blue spikes have two choices, and that this choice can be propagated if each blue spike is uniquely covered, also holds for the other gadgets.

More formally, the property needed in the proof of correctness is that if a blue spike is set to the true state, i.e., it is covered by the adjacent yellow spike above it, then every blue spike below it is also covered by the adjacent yellow spike above it. Symmetrically, we need that if a blue spike is covered by the yellow spike below it, each blue spike above it is covered by the yellow spike below it, too.

Note that in Figure~\ref{fig:wire-appendix} as well as in the figures for the other gadets we connected the base curves to give a small example consisting of connected curves, but these connections are not part of the respective gadget itself. 
The overall construction of $P$ and $Q$ is presented in Section~\ref{sec:appendix-building}.

	\begin{figure}[ht]
		\centering
		\includegraphics[scale=0.8]{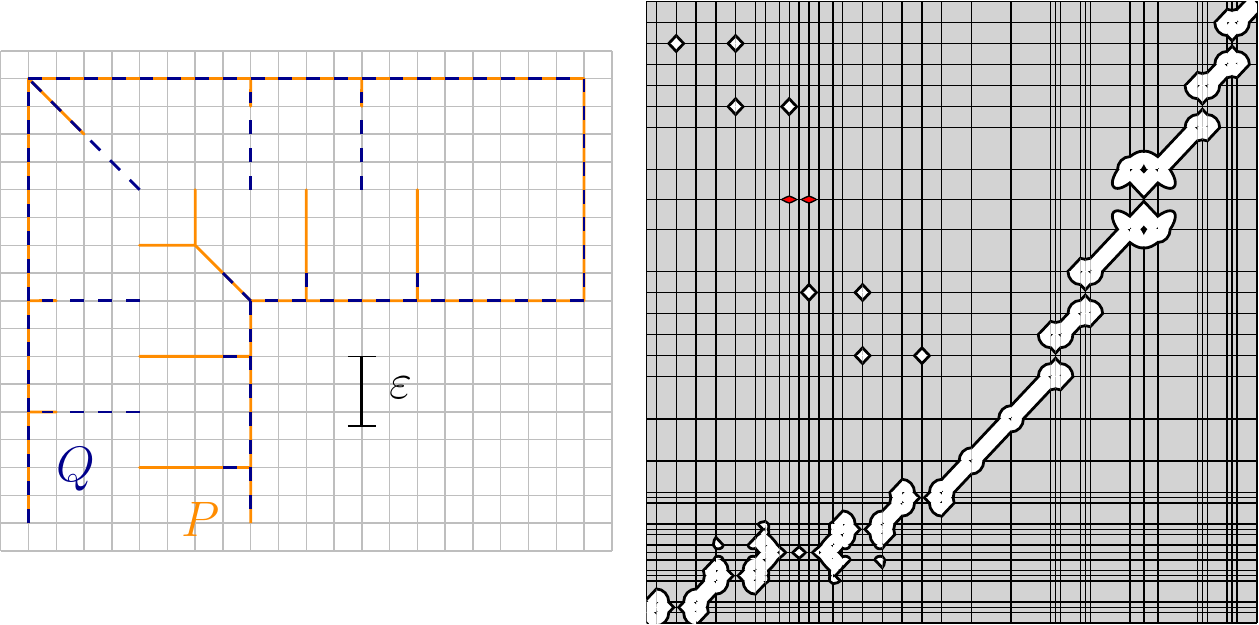}
		\caption{Bend gadget drawn on a (coarse) grid with its resulting free space diagram.}\label{fig:bend}
	\end{figure}

\textbf{Bend gadget.} Next we want to build 90 degree turns in order to ensure that we can form curves to represent the rectilinear edges of the input graph. We simply bend a wire to a right angle and insert a long diagonal spike at the outer corner as well as a ``fork'' with two spikes at the inner one, as displayed in Figure~\ref{fig:bend}. We only need to build right angles since the input instance is, as mentioned, rectilinear.

The free space shows the two staircases connected by a row with two parallel and slightly smaller components to choose from. These are generated by the two spikes of the fork (the $Y$-shaped part of the curve at the inner corner of the gadget).
The bend gadget ensures that the choice made for the wire on one side (either selecting the right or the left of the two components of each row, i.e., covering a blue spike with its upper/right or lower/left partner) propagates to the wire on the other side. The spikes in a bend have the same properties as the spikes in the wire gadget.

	\begin{figure}[ht]
		\centering
		\includegraphics[scale=0.8]{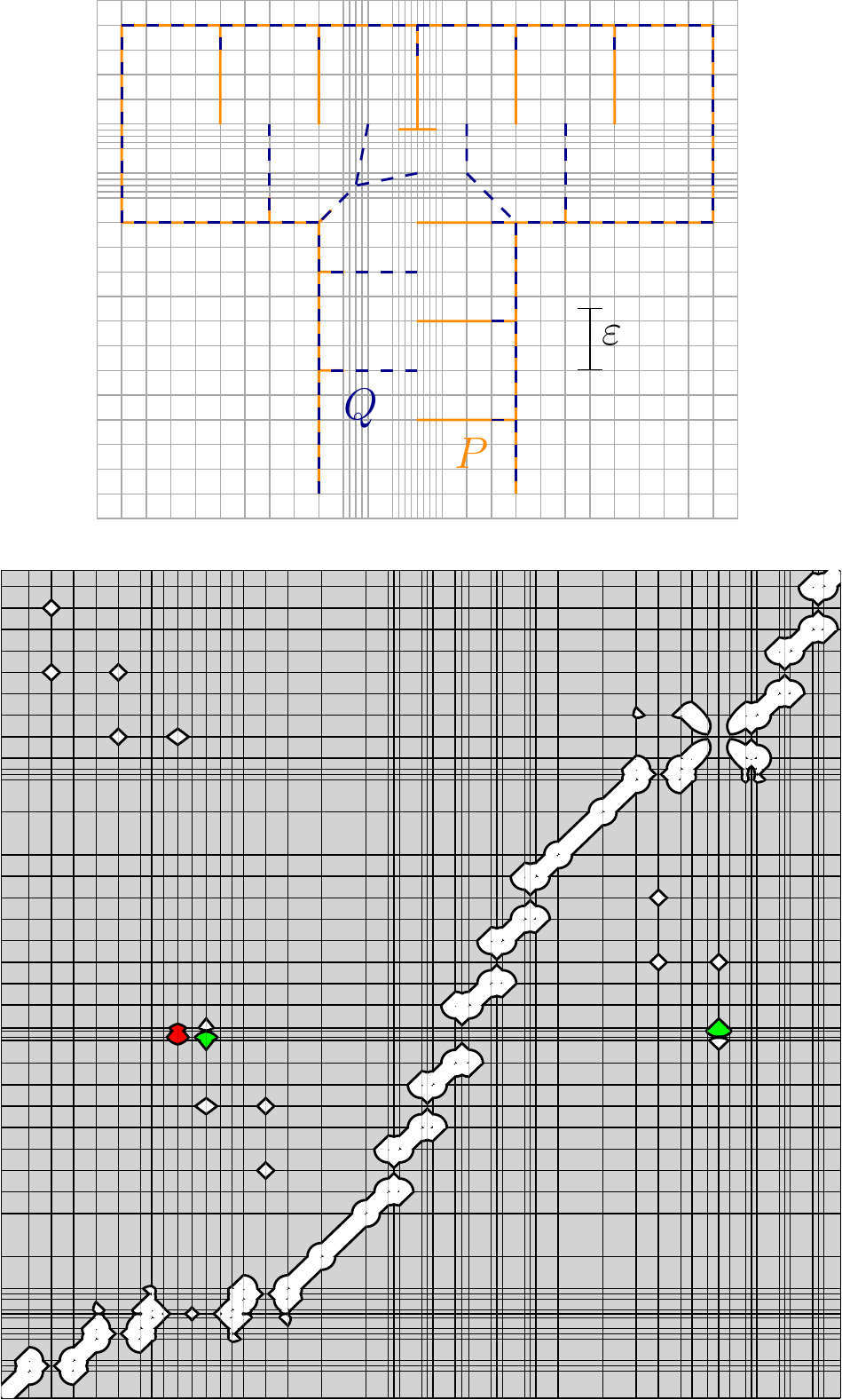}
		\caption{Split gadget drawn on a grid with its resulting free space diagram.}\label{fig:split}
	\end{figure}

\textbf{Split gadget.} Now we focus on how to split a wire into two. In Figure~\ref{fig:split} the gadget is displayed, as well as its free space diagram. Note that the split is probably the most intricate gadget to draw, which is why we drew the underlying grid in its actual density for specific parts of the gadget (and used the coarser grid where only every fourth line is drawn for the rest of the gadget).
The $T$-shaped yellow spike in the middle of the split is called the \emph{split spike}.
We can identify three staircases as we have three wires connecting in the gadget. The top left staircase corresponds to the bottom wire, the right staircase corresponds to the right wire and the bottom left staircase stems from the left wire. Note that splits are directed in the sense that we can only enter a gadget through one specific wire, the \emph{entry wire}. In case of Figure~\ref{fig:split} the entry wire is the bottom one, while the other two wires are called \emph{exit wires} of the split. We discuss how to realize splits that allow exit wires facing straight instead of sideways in Section~\ref{sec:appendix-building}.

In the free space diagram, the split itself consists of the small, colored 
components in the rows just between the second and third staircases: we have two components at its rightmost column and three more to their left in the same row: a larger, red connected component as well as another two small ones. Note that the tiny components next to the green ones are never chosen; they are, among few others in this figure, avoidable components. The red component covers the same interval on the horizontal axis as the ending of the top staircase, the smaller ones are part of the other two staircases. In order to cover the split row (i.e., in order to cover the $T$-shaped spike of curve $P$) we either have to choose the red component, which determines all choices for the three staircases, or select two of the smaller ones, namely those marked green.

Formally this means the following: The split spike is adjacent to three blue spikes. Its $T$-shaped tip can be (fully) covered by either the top blue spike of the entry wire (selecting the red component in the free space), or by both blue spikes of the exit wires (selecting both green components). It cannot be fully covered by a single blue spike from an exit wire.
This implies that we can propagate choices.
Similar to what happens in wires, if we require that each blue spike is uniquely covered, we obtain the following:
If a blue spike in one of the exit wires is covered by the adjacent spike that is further away from the split spike, each blue spike in the entry wire is covered by the yellow spike above it. If a blue spike in the entry wire is covered by the adjacent yellow spike below it, each blue spike in the exit wire is covered by the adjacent spike that is closest to the split spike. Otherwise not all yellow spikes would be covered.


	\begin{figure}[ht]
		\centering
		\includegraphics[scale=0.8]{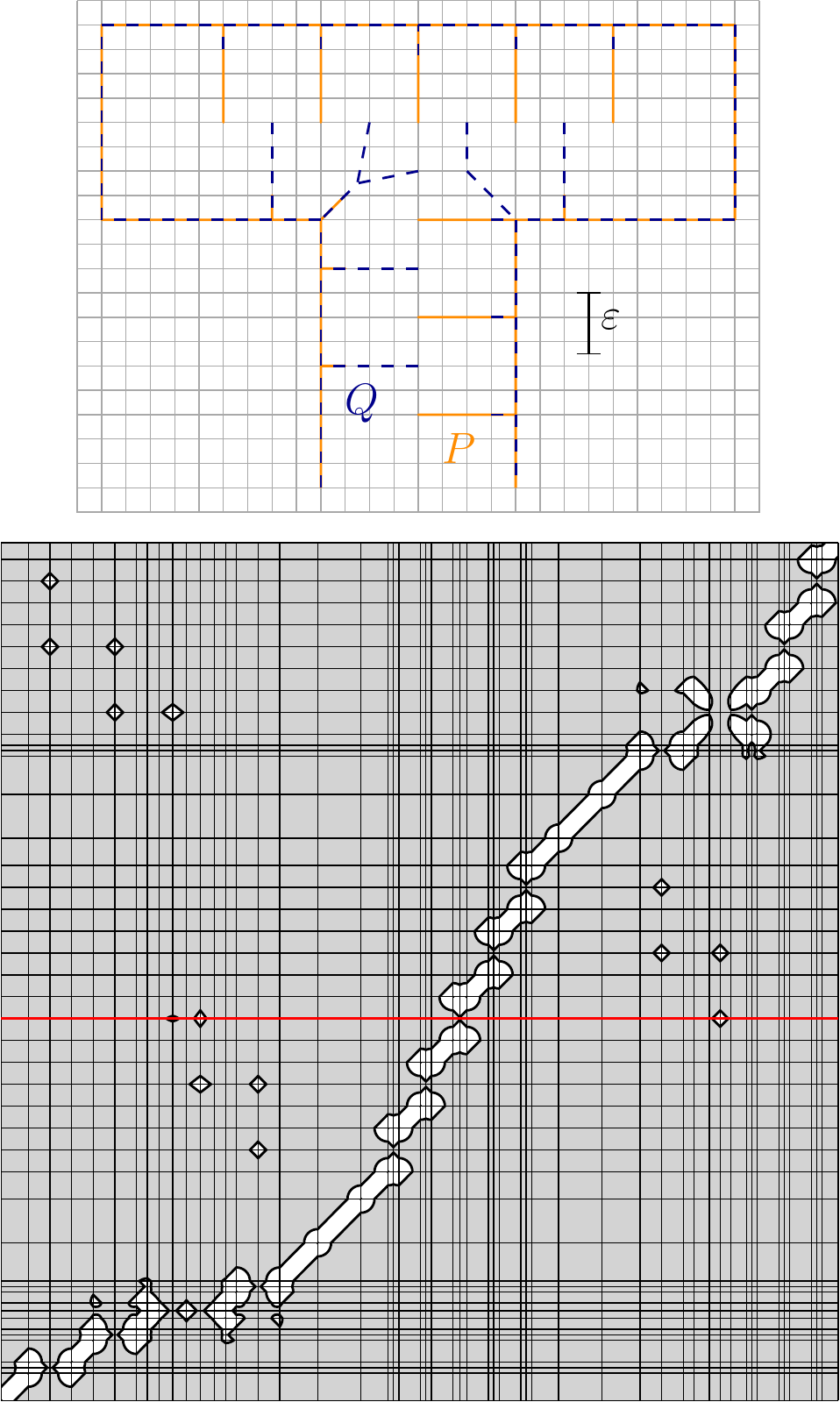}
		\caption{Clause gadget drawn on a (coarse) grid with its corresponding free space diagram. The horizontal red line that corresponds to the tip of the clause spike intersects three components.}\label{fig:clause}
	\end{figure}

\textbf{Clause gadget.}

The clause gadget looks similar to the split gadget, the $T$-shaped spike is simply replaced by a regular spike, which we call \emph{clause spike}. Note again, that we use a coarse grid to indicate the placement of the vertices; for the $Y$-shaped blue spikes we refer to Figure~\ref{fig:split} where (around the intricate parts of the gadget) all grid lines are drawn.

The three staircases in the free space diagram of Figure~\ref{fig:clause} correspond to the same wire parts of the curves as for the split gadget. As a result of omitting the $T$-shape we get one interval, corresponding to the clause spike, that can be covered by three components. We marked the center of the interval in red in Figure~\ref{fig:clause}. Now, in order to fulfill this clause and cover both parameter spaces we need to select at least one of the three components to cover this interval.

When looking at the curves this means that in at least one of the three wires the blue spikes all have to be covered by the adjacent yellow spike that is closest to the clause spike.

Still we need to make sure we can connect all described gadgets to curves representing any input instance, i.e., we have to connect them according to a given a rectilinear embedding of monotone planar 3-SAT formula. 
The next section contains the descriptions of further gadgets we need, while in Section~\ref{sec:appendix-building} we explain how to place the gadgets, connect them consistently and finally traverse the curves.

\subsubsection{Construction gadgets}

\textbf{Color gadget.} First of all, we present a gadget to switch colors of spikes, i.e., a gadget that allows us to change whether the longer spikes of $P$ occur on the outer or inner string of a wire (of course, this also determines the side on which the longer spikes of $Q$ are located). It also changes whether we select the upper or the lower staircase components in the free space diagram.
In entering a gadget we might require a certain state of the spikes, while for the next gadget we need the other one, so changing the states may be necessary.

	\begin{figure}[ht]
		\centering
		\includegraphics[scale=0.9]{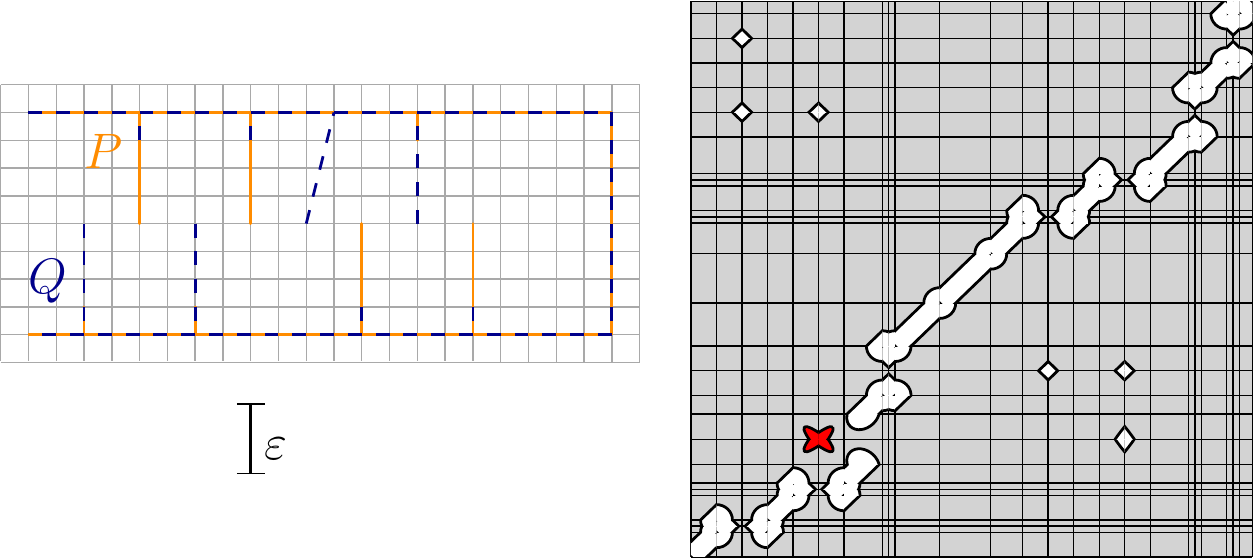}
		\caption{Color gadget drawn on a (coarse) grid with its resulting free space diagram.}\label{fig:col}
	\end{figure}

As displayed in Figure~\ref{fig:col}, the red, $x$-shaped component, which corresponds to the slightly tilted spike being within $\eps$-distance of the upper and the lower long yellow spike, connects the two staircases that stem from the wires to the sides of it.
So, concerning coverage the tilted spike has the same properties as a normal spike.

	\begin{figure}[htbp]
		\centering
		\includegraphics[scale=0.8]{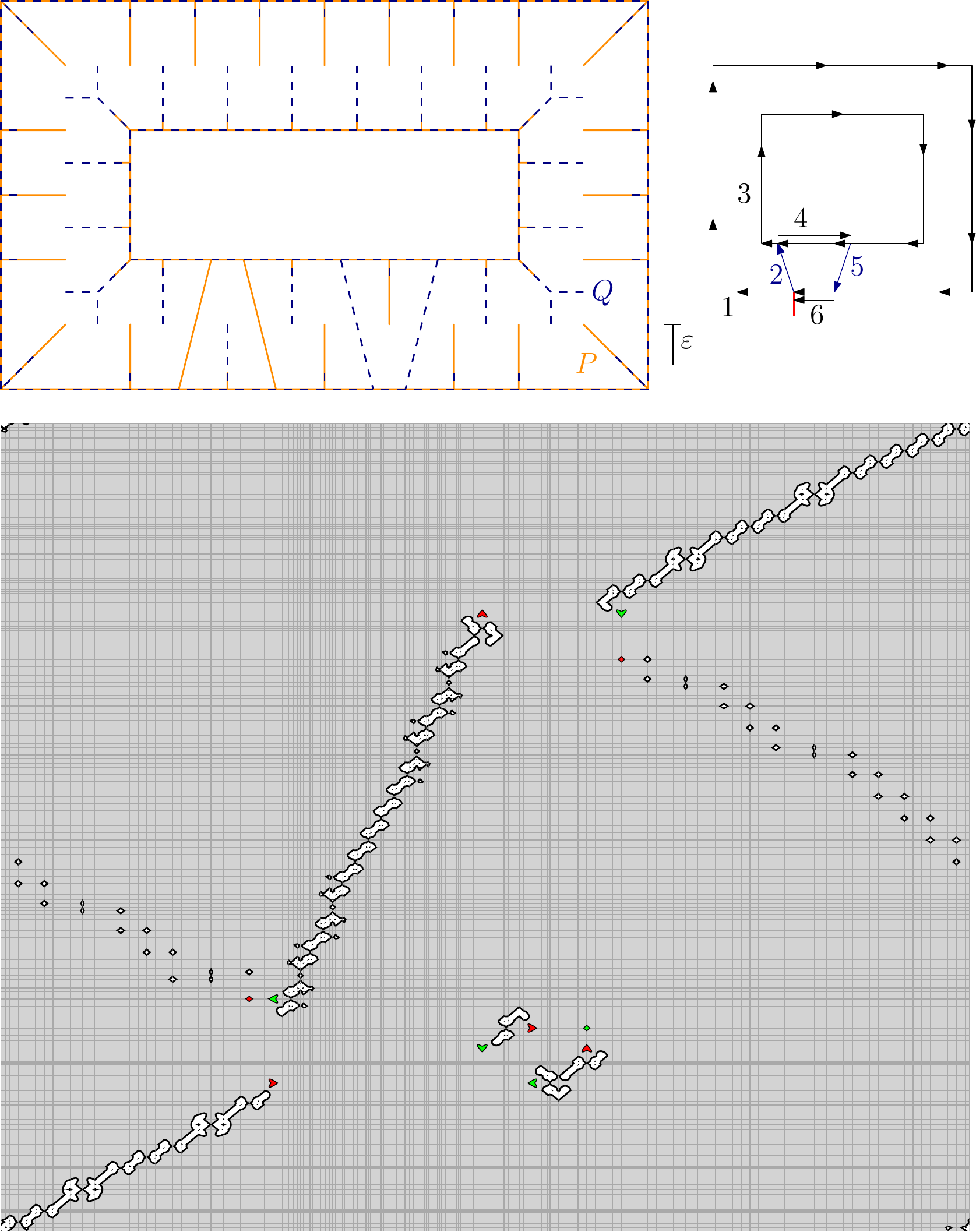}
		\caption{Connection gadget with its resulting free space diagram and traversal order of $P$.}\label{fig:connection}
	\end{figure}

\textbf{Connection gadget.} Now, in order to build curves consisting of two parallel base curves (intersected with spikes, of course) per edge segment of the input graph, we need a gadget which we call connection gadget.
Figure~\ref{fig:connection} shows how the connection gadget works: we build a face of a graph using wires and four bends (of course, this is an artificial face which would not occur in any actual input instance, but it serves its purpose here). To connect the inner and the outer base curves, respectively, we add two tilted segments between the strings.

Figure~\ref{fig:connection} displays the gadget. Here, we omit the underlying grid completely, because it would be rather confusing than helpful. A traversal of the curve $Q$ is drawn next to the curves; the red line indicates beginning and end point of $Q$. Note that the connection of $Q$ is tilted ``downwards'' whereas the other one is tilted ``upwards'', and that each of the tilted segments is only traversed once. The wires and bends produce nice staircases, which can be seen in the resulting free space diagram. Between the two staircases, we have the connection gadget itself. 
Note that the middle of each tilted segment has the same properties as the tip of a spike when it comes to finding a minimal covering selection: each middle point of a tilted segment can only be covered by the adjacent spikes. The middle points are placed on grid points where regular spike tips of a wire would occur.

\textbf{Scissor gadget.} Up to now we have only ensured that we can work with closed curves for our reduction since any distinct beginning or ending of the curves breaks the staircase structures we need. However, it is possible to also work with non-closed curves. The necessary construction is shown in Figure~\ref{fig:scissor}.

This gadget cuts a closed curve open without breaking the staircase structures of the wire where it is located. Both curves start at the upper tip of the $V$-extension on the left hand side of the wire and end at the lower tip of the $V$. Beginning and end just add one larger and one tiny avoidable clutter component each, located in the top left and bottom right corner of the free space diagram. Neither of the four arising components are chosen for a covering selection since the clutter of the wire gadgets already covers the corresponding intervals.

	\begin{figure}[ht]
		\centering
		\includegraphics[scale=0.9]{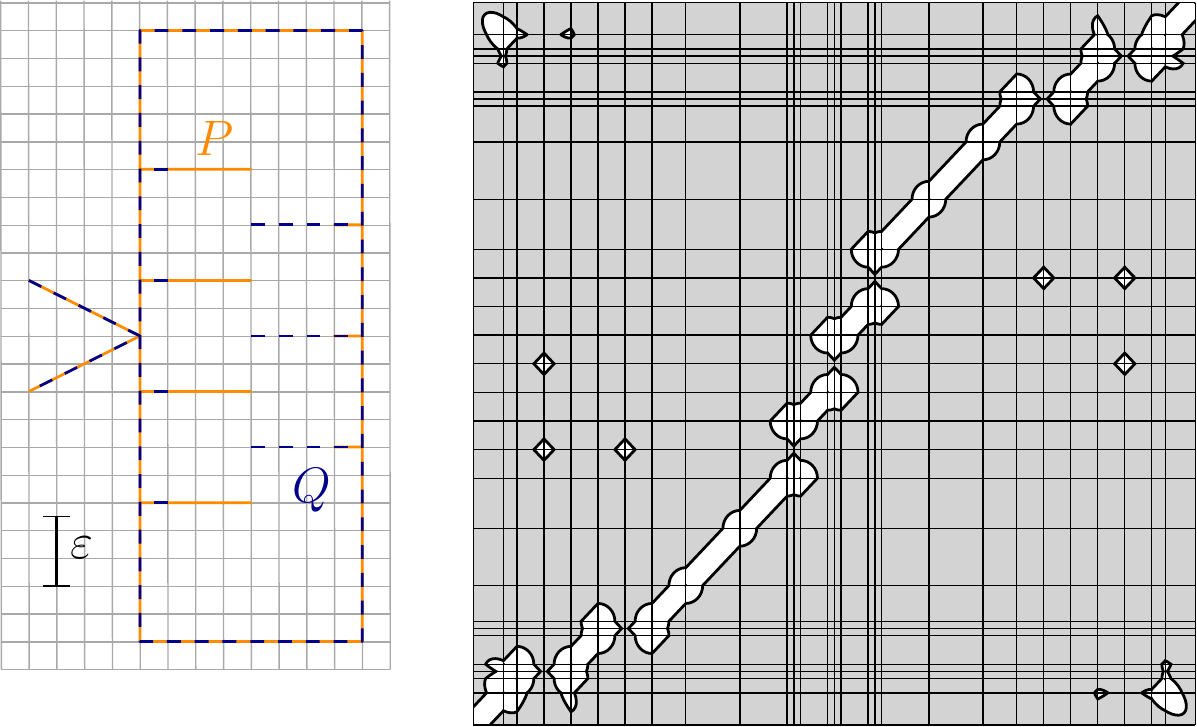}
		\caption{Scissor gadget drawn on a (coarse) grid with its resulting free space diagram.}\label{fig:scissor}
	\end{figure}

\subsection{Building the curves}\label{sec:appendix-building}

Our input instance is the drawing of a given 3-SAT instance, i.e., a rectilinear graph where each vertex has at most degree three. We first need to perform some transformations and scaling to make sure that we can draw curves along the edges and vertices of the graph.

\begin{lemma}
We can place curves consisting of the described gadgets along the edges and vertices of the input graph $G$.
\end{lemma}

Recall that the graph can be drawn on a grid. We set the length of a shortest segment (that is a segment of an edge between a 90° turn and a vertex or between two vertices) to be the grid size, i.e., the unit length.
In order to separate the vertices and/or bends from one another, we first scale the graph by a factor of 2.

\textbf{Dealing with 2-clauses.} For clauses consisting of only two literals, we simply double one of the participating literals by adding a split before the respective wire arrives at the clause (e.g., we replace the clause $(v_i \lor v_j)$ with $(v_i \lor v_i \lor v_j)$). Since we scaled by a factor of 2 beforehand, we ensured that we have a ``box'' of $2\times 2$ grid cells around each vertex inside which no other vertex is placed. We scale by 2 again and add splits as described above.

	\begin{figure}[ht]
		\centering
		\includegraphics[scale=0.9]{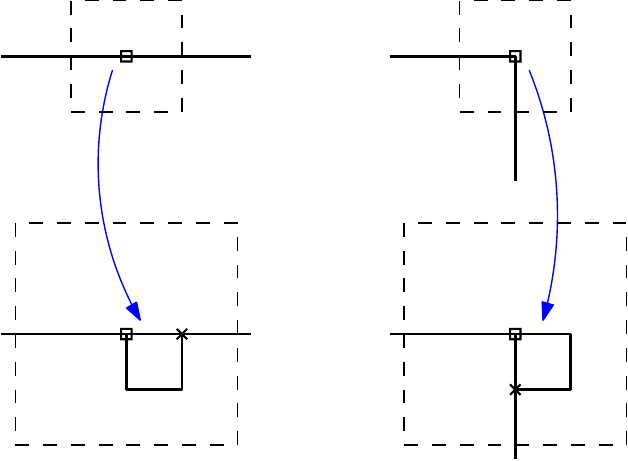}
		\caption{Possible 2-clauses and the inserted splits.}\label{fig:clausewaround}
	\end{figure}

In Figure~\ref{fig:clausewaround}, clause vertices are displayed as boxes and splits as crosses. The split insertion depends on whether the clause is entered from opposite sides or from neighboring sides. For the first case, we insert a split of the vertical wire at unit length distance from the clause vertex, for the latter case we place the split at either wire. Mirroring ensures that all possible cases are dealt with.
Note that we can place the split safely due to the scaling: we ensured that no other vertex is placed within $4 \times 4$ grid cells around the clause, such that we can place the split without interfering with any other vertex or edge of the graph.

\textbf{Dealing with directed splits.} Now, recall that our split gadget is directed and only allows to split a vertical wire into two horizontal wires. However, if a variable takes part in more than two clauses of the same type (positive or negative) or if the respective clauses have to be drawn at a different height in the input graph, we need to ensure we can exit a split vertically. In order to do so, we simply bend our entry wire three times such that it is facing the same direction as one of the desired exit wires, namely the horizontal one. 

	\begin{figure}[ht]
		\centering
		\includegraphics[scale=0.7]{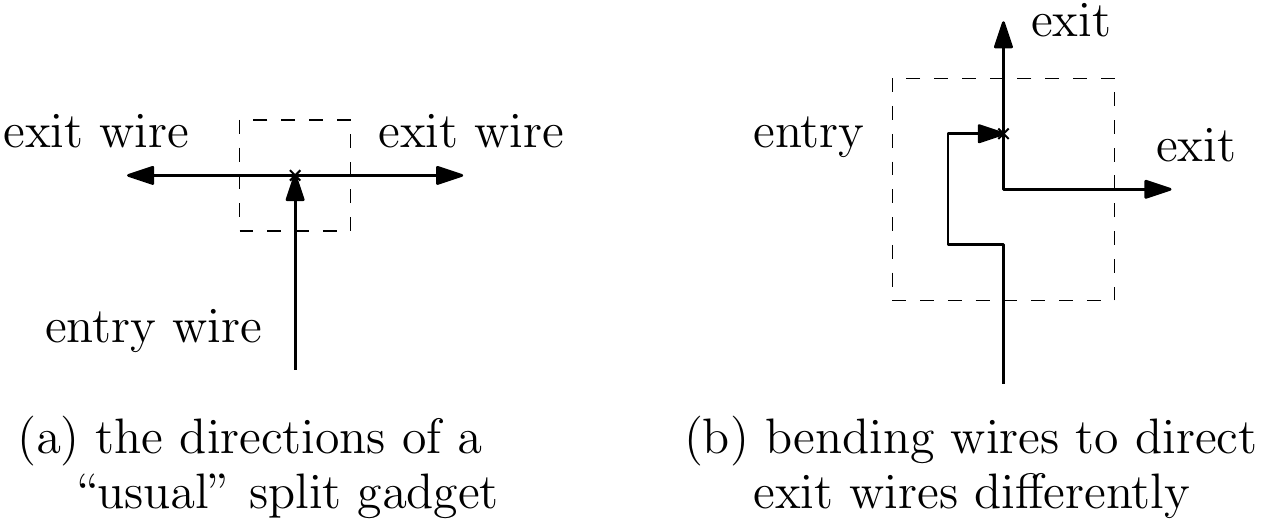}
		\caption{Bending wires to make exit wires of a split face the desired directions.}\label{fig:splitwaround}
	\end{figure}

As a result we have two vertical exit wires of which we bend the one facing downwards once, which ensures that it is then facing the correct direction. To place the bends safely, we scale again by factor 2. This is visualized in Figure~\ref{fig:splitwaround}. Note once more that the construction can be mirrored to enable all possible exit directions.

\textbf{Placing connection gadgets on wires.} Finally, we want to ensure that we can connect the inner and outer base curves of our wires at any segment. We picture to place the wires such that the tips of the long spikes would touch the edges of the input graph, i.e., the edges lie parallel to and midway in between the two base curves. To place the connection gadget, we need to have enough width to place all four tilted connection segments, two per curve $P$ and $Q$. Counting from the center of the bends (long diagonal yellow spikes) we need to be able to place 8 blue spikes (which can be counted by looking at the parallel wire above the connection gadget in Figure~\ref{fig:connection}). To place the vertices of all gadgets on grid points, we asked that two spikes of the same color should be 16 times the unit length apart. As a result we have to scale once more by factor $8\cdot 16 = 128 = 2^7$.
In total, we are scaling the input graph by a factor of $2^{10}=1024$.

\textbf{Completing the construction.} As a last step we describe how to concatenate the gadgets to build our curves.
Recall that without the connection gadget each face of the graph $G$ contains one closed curve. Now we describe in which order to connect the faces to each other.
First of all, we consider the dual graph $G'$ of the input graph $G$, i.e., we draw a second graph with vertices $v_i$ for every face $f_i$ of $G$ and edges between vertices whenever the corresponding faces are adjacent. Recall that there is also a vertex for the outer face; we call it $v$. We compute a spanning tree of $G'$ by deleting edges. As an example we used the instance of Figure~\ref{fig:input} and drew $G'$ on top, see Figure~\ref{fig:tree} below. The blue edges represent a spanning tree of $G'$, the dashed ones are the removed edges.

	\begin{figure}[ht]
		\centering
		\includegraphics[scale=0.6]{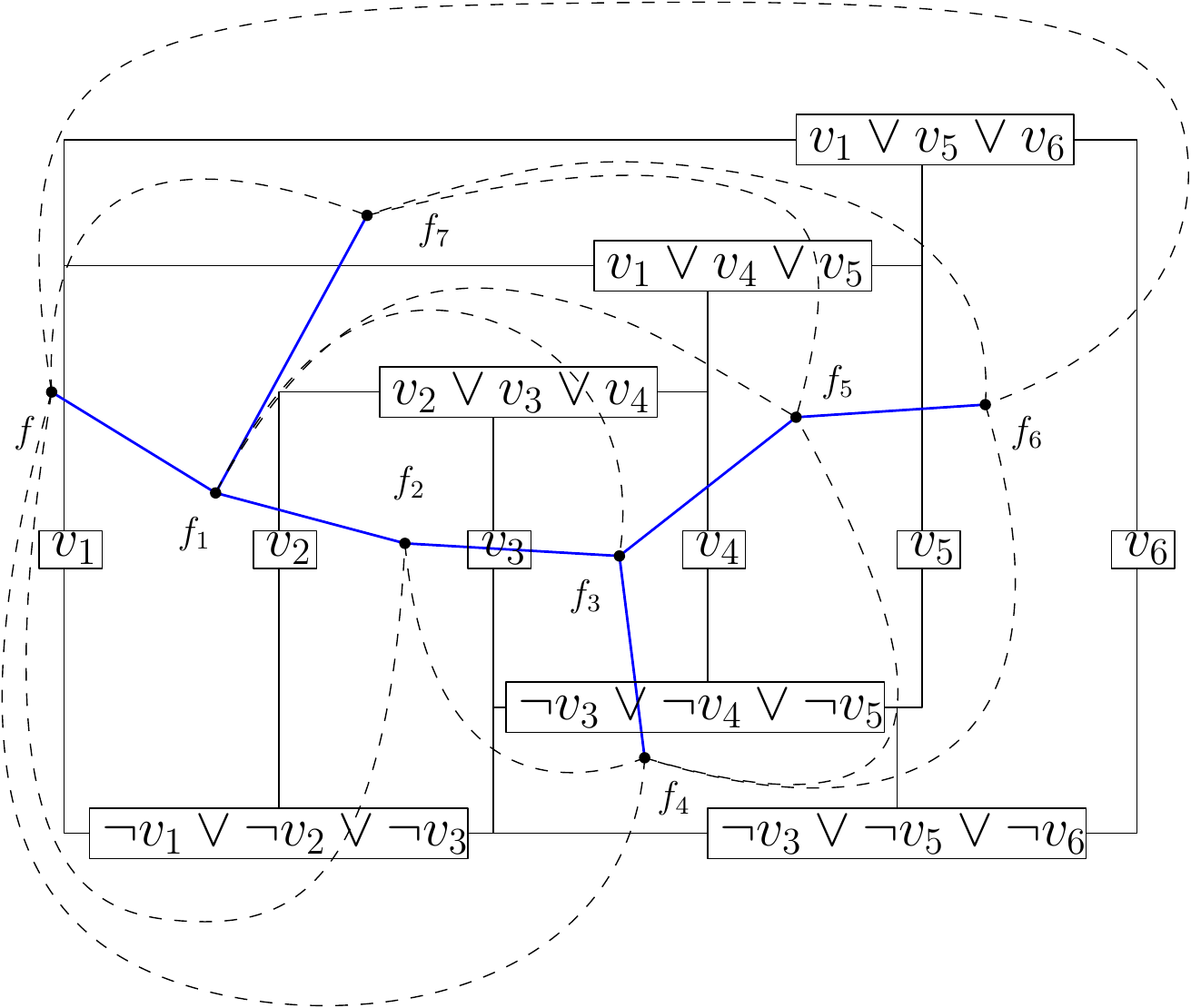}
		\caption{The input graph $G$, its dual $G'$ and its minimal spanning tree.}\label{fig:tree}
	\end{figure}

We build our curves as follows: if the
vertex $v$ is a leaf of the spanning tree, just as in Figure~\ref{fig:tree},  we can choose $v$ to be our root of the tree. Starting at the root $v$ we traverse the tree, say, in clockwise order, until we end at $v$ again. Let the tour be $\langle v, v_1, \ldots , v_i, v_1,\ldots , v_m\rangle$ and correspond to the faces $f$ (the outer face) and $f_1, \ldots , f_m$ (inner faces).  Note that we traverse each vertex $v_i$ exactly $deg(v_i)$ times.
Now we build our curves $P$ and $Q$ as follows: first, we start on the outer base curve adjacent to the first inner face $f_1$ and completely draw the outer base curve. Back at the starting point we use the first half of a connection gadget to enter the inner base curve of the first face. We traverse it until we are adjacent to the next face, say $f_7$, which we then enter through a connection gadget. Note that we only drew the part of the inner base curve of $f_1$ between entering from the outer base curve and exiting to the next face. When reentering $f_1$ we continue (clockwise) to draw the inner base curve.  The procedure is displayed in Figure~\ref{fig:trav}. Note that we used the instance from Figure~\ref{fig:input} again, but omitted the labeling: we marked variables as disks, clauses as squares and splits as crosses instead.

	\begin{figure}[ht]
		\centering
		\includegraphics[scale=0.6]{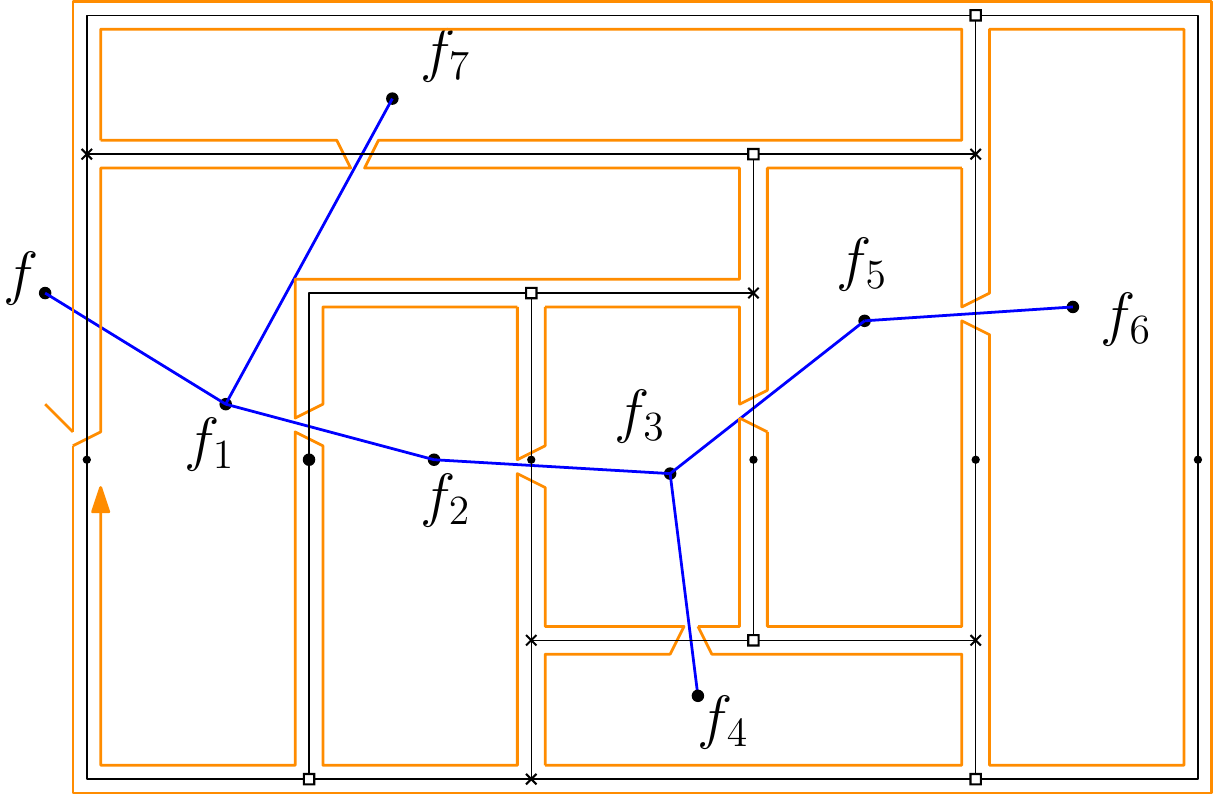}
		\caption{Building the curve $P$ corresponding to an input instance.}\label{fig:trav}
	\end{figure}

If the vertex $v$ is not a leaf but an inner vertex of the tree, we cannot traverse the whole outer base curve first. Instead, we label the vertices adjacent to $v$ clockwise, say $v_1, \ldots, v_h$ and traverse parts of the outer base curve to connect between traversals of inner faces. We start at the outer face and traverse a small part of it along the inner face $f_1$ (the first one by our ordering of the adjacent vertices). We enter this face through the connection gadget and follow the inner faces according to the subtree rooted at $f_1$. Just as explained above, we exit $f_1$ where we entered after completing the traversal of its subtree. We continue drawing the outer base curve from there on until we are adjacent to the next inner face $f_2$, which we enter again. Just as before we traverse its subtree, and so on. Following the tree in this manner results in a complete construction of the curves as we traverse the graph.

\subsection{Analysis}\label{sec:appendix-analysis}
We still need to formally prove correctness of our construction, i.e., give the detailed proof for Theorem~\ref{thm:np}, which states that deciding the $k$-Fréchet distance is NP-hard.
Recall that $k=k_b+k_c$, where $k_b$ is the number of blue spikes and $k_c$ is the number of clutter components. Notably, this means that each interval corresponding to a blue spike can only be covered by two different components.

\begin{proof}
\textbf{'$\Leftarrow$'}
Let $g:V\rightarrow \{true, false\}$ be an assignment of the variables $v \in V$ that satisfies the 3-SAT formula.
Now we explain how to cover the parameter spaces with $k$ components.
First, for each interval of the parameter spaces that can only be covered by a single component, we choose that component.
After this step we have chosen all the (necessary) clutter.
It follows that only spikes remain to be covered.
Then, for each variable $v$, if $g(v)$ is true, we cover the central spike of the corresponding wire by the adjacent yellow spike $y$ above it. We then propagate this choice: the blue spike $b$ above that yellow spike $y$ is covered by the yellow spike above $b$ and so on. For each variable set to false we do the inverse, i.e., we cover the central spike by the yellow spike below it and then propagate.

This selection of boxes has the effect that each blue box is only covered once, therefore the number of chosen boxes is precisely $k=k_c+ k_b$.
Furthermore, since $g$ is an assignment that satisfies the formula, we have that for each clause gadget the yellow clause spike is covered by at least one blue spike.
This means that both parameter spaces are completely covered by our selection.

\textbf{'$\Rightarrow$'}
In the following, we assume that we have a selection of boxes $S$ with $|S|= k$ that covers the parameter spaces.
First we know that $S$ contains all the clutter, because otherwise it would not be covering.
Now let $R$ equal $S$ minus the clutter. Due to the fact that $k=k_c + k_b$, we know that $|R|= k_b$.

We know that no spike can be covered by clutter and for each blue spike $b$ there are exactly two possible components that can cover the interval corresponding to $b$. Additionally, no component can cover more than one blue spike. It follows that $R$ contains exactly one component per blue spike, so each blue spike is covered by exactly one yellow spike.

We now define an assignment of the variables $g:V\rightarrow \{true, false\}$ as follows: 
For each variable $v$, if there is a component in $R$ that corresponds to the central spike of $v$ and the yellow spike above it, we set $g(v)$ to true. Else, i.e., if the central spike is covered by the lower adjacent spike, we set $g(v)$ to false.

Now we need to prove that $g$ actually satisfies the 3-SAT formula.
We do this by proving that each clause is satisfied.
We know that each clause spike is covered, so there is a corresponding component in $R$, and we know that each clause spike can only be covered by (at least) one of the three adjacent blue spikes. 
Let $c$ be a clause and let $b(c)$ be a blue spike that is covered by the (yellow) clause spike $y(c)$ of $c$. Now, let $v$ be the variable that corresponds to the wire of $b(c)$.

We define the \emph{staircase distance} of two spikes $a, b$ of a variable $v$ to be the number of spikes in between them plus one.
Two adjacent spikes have staircase distance 1, two spikes of the same color have at least staircase distance 2. A spike $a$ is said to be \emph{closer} to a spike $b$ (within the same variable) than a third spike $c$ if the staircase distance of $a$ and $b$ is smaller than the distance of $b$ and $c$.
Intuitively, the staircase distance counts number of ``steps'' between two spikes on the staircase corresponding to a wire in the freespace.

We prove by induction that the (blue) central spike $b(v)$ of $v$ is covered by the yellow spike that is closer to $y(c)$, which, due to the monotonicity of the formula, means that $g(v)$ is set such that $c$ is fulfilled.
Let $2j$ be the distance between $b(c)$ and $b(v)$. The distance is even because they are both blue spikes.
For $i\in \{0,\dots,j\}$, let $b_i$ be the (blue) spike between $b(c)$ and $b(v)$ whose distance to $b(c)$ is $2i$.
The induction hypothesis is that $b_i$ is covered by the adjacent yellow spike that is closer to the clause spike $y(c)$.
For $i=0$ this is true, because $b_0=b(c)$.

Now assume that the hypothesis holds for any $i\in \{1,\dots,j\}$. The spike $b_i$ is covered by the adjacent yellow spike closer to $y(c)$, which is also the one further away from $c$ from the viewpoint of $b_{i-1}$.
Now let $y$ be the yellow spike between $b_i$ and $b_{i-1}$. If $y$ is simply part of the wire, it can only be covered by one of the two adjacent blue spikes and we know that, since we only cover each blue spike once, choices are propagated.
If $y$ is a split spike, we know that it can either be covered by the spike within the entry wire, which is $b_{i-1}$, or $y$ can be covered by both other adjacent spikes simultaneously, i.e., the spikes that are within the exit wires.
Now, due to $b_i$ not being covered by $y$, we have that $b_{i-1}$ is covered by $y$, thus finishing the proof by induction.
Since $b_j=b(v)$, we have that the central spike is covered in such a way that the induced assignment value $g(v)$ of $v$ fulfills $c$.
\end{proof}

\end{document}